\pdfoutput=1 
\newif\iffocs

\iffocs
  \documentclass[10pt, conference, compsocconf]{IEEEtran}
\else
  \documentclass[11pt]{article}
  \usepackage[margin=1in]{geometry}
\fi

\usepackage{amsmath,amsfonts,amssymb,amsthm}
\usepackage{mathtools}
\usepackage{thm-restate}
\usepackage{tikz,tikz-qtree}
\usepackage[noend,ruled]{algorithm2e}
\usepackage{comment}

\usepackage{mmap}
\usepackage[T1]{fontenc}

\usepackage[utf8]{inputenc}
\usepackage{csquotes}
\usepackage[english]{babel}

\usepackage[
    activate={true,nocompatibility},
    final, 
    tracking=true,
    kerning=true,
    spacing=true,
    factor=1100,
    stretch=10,
    shrink=10
    ]{microtype}
\microtypecontext{spacing=nonfrench}

\iffocs
\usepackage[backend=biber,
              style=alphabetic,
        maxbibnames=99,
      minalphanames=3,
      maxalphanames=4,
            backref=false]{biblatex}
\else
\usepackage[backend=biber,
              style=alphabetic,
        maxbibnames=99,
      minalphanames=3,
      maxalphanames=4,
            backref=true]{biblatex}
\fi

\usepackage[hidelinks]{hyperref}
\hypersetup{
    pdftitle={A Tight Composition Theorem for the Randomized Query Complexity of Partial Functions}, 
    pdfauthor={Shalev Ben{-}David, Eric Blais}, 
    colorlinks=true, 
    linkcolor=blue, 
    citecolor=blue, 
    urlcolor=green 
}

\AtEveryBibitem{
 \clearlist{address}
 \clearfield{date}
 \clearlist{location}
 \clearfield{month}
 \clearfield{series}
 \clearfield{pages}
 \clearlist{organization}
 \clearfield{number}
 \clearlist{language}

 \ifentrytype{book}{}{
  \clearlist{publisher}
  \clearname{editor}
  \clearfield{issn}
  \clearfield{isbn}
  \clearfield{volume}
 }
}

\DeclareFieldFormat{eprint:eccc}{
\mkbibacro{ECCC}\addcolon\space\ifhyperref
    {\href{http://eccc.hpi-web.de/report/#1/}{\nolinkurl{#1}}}
    {\nolinkurl{#1}}}
\DeclareFieldAlias{eprint:ECCC}{eprint:eccc}
\DeclareFieldFormat{eprint:iacr}{Cryptology ePrint\addcolon\space\ifhyperref
    {\href{https://eprint.iacr.org/#1}{\nolinkurl{#1}}}
    {\nolinkurl{#1}}}
\DeclareFieldAlias{eprint:IACR}{eprint:iacr}
\DeclareFieldFormat{eprint:arxiv}{arXiv\addcolon\space\ifhyperref
    {\href{https://arxiv.org/abs/#1}{\nolinkurl{#1}}}
    {\nolinkurl{#1}}}

\DeclareFieldFormat[article]{title}{#1}
\DeclareFieldFormat[inproceedings]{title}{#1}
\DeclareFieldFormat[incollection]{title}{#1}
\DeclareFieldFormat[online]{title}{#1}

\renewbibmacro{in:}{}

\newcommand{\lName}{1}

\newcommand{\donothing}[1]{#1}

\newcommand{\JACM}{\if\lName1\donothing{Journal of the {ACM}}\else{JACM}\fi}
\newcommand{\SICOMP}{\if\lName1\donothing{{SIAM} Journal on Computing}\else{SICOMP}\fi}
\newcommand{\ToC}{\if\lName1\donothing{Theory of Computing}\else{ToC}\fi}
\newcommand{\ToCGS}{\if\lName1\donothing{Theory of Computing Graduate Surveys}\else{ToC}\fi}
\newcommand{\TOCT}{\if\lName1\donothing{{ACM} Transactions on Computation Theory}\else{TOCT}\fi}
\newcommand{\CCjournal}{\if\lName1\donothing{Computational Complexity}\else{CC}\fi}
\newcommand{\CJTCS}{\if\lName1\donothing{Chicago Journal of Theoretical Computer Science}\else{CJTCS}\fi}

\newcommand{\TCS}{\if\lName1\donothing{Theoretical Computer Science}\else{TCS}\fi}
\newcommand{\IPL}{\if\lName1\donothing{Information Processing Letters}\else{IPL}\fi}
\newcommand{\JCSS}{\if\lName1\donothing{Journal of Computer and System Sciences}\else{JCSS}\fi}

\newcommand{\RSA}{\if\lName1\donothing{Random Structures and Algorithms}\else{RSA}\fi}
\newcommand{\JCTA}{\if\lName1\donothing{Journal of Combinatorial Theory, Series A}\else{JCTA}\fi}
\newcommand{\JCTB}{\if\lName1\donothing{Journal of Combinatorial Theory, Series B}\else{JCTB}\fi}
\newcommand{\PJM}{\if\lName1\donothing{Pacific Journal of Mathematics}\else{PJM}\fi}
\newcommand{\QIC}{\if\lName1\donothing{Quantum Information and Computation}\else{QIC}\fi}
\newcommand{\IJQI}{\if\lName1\donothing{International Journal of Quantum Information}\else{IJQI}\fi}

\DefineBibliographyStrings{english}{%
  backrefpage = {p{.}},
  backrefpages = {pp{.}},
}

\newtheorem{theorem}{Theorem}
\newtheorem{lemma}[theorem]{Lemma}

\newtheorem{corollary}[theorem]{Corollary}
\newtheorem{definition}[theorem]{Definition}
\newtheorem{conjecture}[theorem]{Conjecture}

\newtheorem{claim}[theorem]{Claim}
\newtheorem*{question}{Main Question}

\theoremstyle{definition}

\newcommand{\be}{\begin{equation}}
\newcommand{\ee}{\end{equation}}
\newcommand{\eq}[1]{\hyperref[eq:#1]{(\ref*{eq:#1})}}
\renewcommand{\sec}[1]{\hyperref[sec:#1]{Section~\ref*{sec:#1}}}
\newcommand{\thm}[1]{\hyperref[thm:#1]{Theorem~\ref*{thm:#1}}}
\newcommand{\lem}[1]{\hyperref[lem:#1]{Lemma~\ref*{lem:#1}}}
\newcommand{\defn}[1]{\hyperref[def:#1]{Definition~\ref*{def:#1}}}
\newcommand{\prop}[1]{\hyperref[prop:#1]{Proposition~\ref*{prop:#1}}}
\newcommand{\cor}[1]{\hyperref[cor:#1]{Corollary~\ref*{cor:#1}}}
\newcommand{\fig}[1]{\hyperref[fig:#1]{Figure~\ref*{fig:#1}}}
\newcommand{\tab}[1]{\hyperref[tab:#1]{Table~\ref*{tab:#1}}}
\newcommand{\alg}[1]{\hyperref[alg:#1]{Algorithm~\ref*{alg:#1}}}
\newcommand{\app}[1]{\hyperref[app:#1]{Appendix~\ref*{app:#1}}}
\newcommand{\conj}[1]{\hyperref[conj:#1]{Conjecture~\ref*{conj:#1}}}
\newcommand{\chap}[1]{\hyperref[chap:#1]{Chapter~\ref*{chap:#1}}}
\newcommand{\clm}[1]{\hyperref[clm:#1]{Claim~\ref*{clm:#1}}}
\newcommand{\fct}[1]{\hyperref[fct:#1]{Fact~\ref*{fct:#1}}}
\newcommand{\qstn}[1]{\hyperref[qstn:#1]{Question~\ref*{qstn:#1}}}
\makeatletter
\newcommand*\rel@kern[1]{\kern#1\dimexpr\macc@kerna}
\newcommand*\widebar[1]{%
  \begingroup
  \def\mathaccent##1##2{%
    \rel@kern{0.8}%
    \overline{\rel@kern{-0.8}\macc@nucleus\rel@kern{0.2}}%
    \rel@kern{-0.2}%
  }%
  \macc@depth\@ne
  \let\math@bgroup\@empty \let\math@egroup\macc@set@skewchar
  \mathsurround\z@ \frozen@everymath{\mathgroup\macc@group\relax}%
  \macc@set@skewchar\relax
  \let\mathaccentV\macc@nested@a
  \macc@nested@a\relax111{#1}%
  \endgroup
}
\makeatother
\renewcommand{\bar}{\widebar}
\newcommand{\eprint}[1]{{\small \upshape \tt \href{http://arxiv.org/abs/#1}{#1}}}

\newcommand{\B}{\{0,1\}}

\newcommand{\AND}{\textsc{AND}}

\newcommand{\Id}{\textsc{Id}}

\DeclareMathOperator{\D}{D}
\DeclareMathOperator{\R}{R}
\DeclareMathOperator{\Q}{Q}

\DeclareMathOperator{\RS}{RS}

\DeclareMathOperator{\bs}{bs}
\DeclareMathOperator{\fbs}{fbs}

\DeclareMathOperator{\s}{s}
\DeclareMathOperator{\supp}{supp}

\DeclareMathOperator{\Dom}{Dom}

\DeclareMathOperator*{\E}{\mathbb{E}}

\DeclareMathOperator{\compR}{compR}
\DeclareMathOperator{\noisyR}{noisyR}

\DeclareMathOperator{\cost}{cost}
\DeclareMathOperator{\height}{height}

\DeclareMathOperator{\tran}{tran}

\DeclareMathOperator{\sfR}{sfR}
\DeclareMathOperator{\Bernoulli}{Bernoulli}
\newcommand{\gapmaj}{\textsc{GapMaj}}

\newcommand{\triv}{\textsc{Triv}}

\DeclareMathOperator{\JS}{JS}
\DeclareMathOperator{\h}{h}
\DeclareMathOperator{\Ess}{S}

\DeclareMathOperator{\argmin}{argmin}

\newcommand{\tTheta}{\widetilde{\Theta}}

\newcommand{\bN}{\mathbb{N}}
\newcommand{\bE}{\mathbb{E}}

\newcommand{\na}{\textsc{na}}

\iffocs
  \excludecomment{fulltext}
\else
  \includecomment{fulltext}
\fi

\begin{document}

\title{A Tight Composition Theorem for the Randomized  
\iffocs
\else
  Query Complexity of Partial Functions
\fi
}

\iffocs
\else
    \author{
    Shalev Ben{-}David\\
    \small University of Waterloo\\
    \small \texttt{shalev.b@uwaterloo.ca}
    \and
    Eric Blais\\
    \small University of Waterloo\\
    \small \texttt{eric.blais@uwaterloo.ca}
    }
\fi

\maketitle

\begin{abstract}
We prove two new results about the randomized
query complexity of composed functions.
First, we show that the randomized composition
conjecture is false: there are families of partial
Boolean functions $f$ and $g$ such that
$\R(f\circ g)\ll\R(f)\R(g)$. In fact,
we show that the left hand side can be polynomially
smaller than the right hand side
(though in our construction,
both sides are polylogarithmic
in the input size of $f$).

Second, we show that for all $f$ and $g$,
$\R(f\circ g)=\Omega(\noisyR(f)\R(g))$,
where $\noisyR(f)$ is a measure describing the
cost of computing $f$ on noisy oracle inputs.
We show that this composition theorem is the
strongest possible of its type:
for any measure $M(\cdot)$ satisfying
$\R(f\circ g)=\Omega(M(f)\R(g))$ for all $f$ and $g$,
it must hold that $\noisyR(f)=\Omega(M(f))$ for all $f$.
We also give a clean characterization
of the measure $\noisyR(f)$: it satisfies
$\noisyR(f)=\Theta(\R(f\circ\gapmaj_n)/\R(\gapmaj_n))$,
where $n$ is the input size of $f$
and $\gapmaj_n$ is the $\sqrt{n}$-gap majority
function on $n$ bits.
\end{abstract}

\iffocs
\else
    \clearpage
    {\small\tableofcontents}
    \clearpage
\fi

\section{Introduction}

In any computational model, one may ask the following
basic question: is computing a function $g$ on $n$ independent inputs
roughly $n$ times as hard as computing $g$ on a single input?
If so, a natural followup question arises: how hard is computing
some function $f\colon\B^n\to\B$ of the value of $g$ on $n$ inputs?
Can this be characterized in terms of the complexity
of the function $f$?

Query complexity is one of the simplest settings 
in which one can study these joint computation
questions. In query complexity, a natural conjecture is that for
any such functions $f$ and $g$, the cost of computing
$f$ on the value of $g$ on $n$ inputs is roughly the cost of computing $f$
times the cost of computing $g$. Indeed, using
$f\circ g$ to denote the composition of $f$ with $n$
copies of $g$, it is known that the deterministic query complexity (also known as the 
decision tree complexity) of composed functions satisfies $\D(f\circ g)=\D(f)\D(g)$
\cite{Tal13,Mon14}.
It is also known that the quantum query complexity (in the bounded-error setting) of 
composed functions satisfies $\Q(f\circ g)=\Theta(\Q(f)\Q(g))$
\cite{Rei11,LMR+11,Kim13}.

However, despite significant interest,
the situation for randomized query complexity
is not well understood, and it is currently unknown
whether $\R(f\circ g)=\tTheta(\R(f)\R(g))$ holds for all Boolean functions
$f$ and $g$. It is known that the upper bound of $\R(f\circ g)=O(\R(f)\R(g)\log\R(f))$ holds.
This follows from running an algorithm for $f$ on the outside,
and then using an algorithm for $g$ to answer each query made by the algorithm
for $f$. (The log factor in the bound is due to the need to amplify
the success probability of the algorithm for $g$ so that it has small error.)
The randomized composition conjecture in query complexity posits that there is a lower bound that 
matches this upper bound up to logarithmic factors;  this conjecture is the focus of the current work.

\begin{question}
Do all Boolean functions $f$ and $g$ satisfy
$\R(f\circ g)=\Omega\big(\R(f)\R(g)\big)$?
\end{question}

Note that there are actually two different versions of this
question, depending on whether $f$ and $g$ are allowed to be
partial functions. A \emph{partial function} is a function
$f\colon S\to\B$ where $S$ is a subset of $\B^n$,
and a randomized algorithm computing it is only required
to be correct on the domain of $f$. (Effectively, the input string
is promised to be inside this domain.)
When composing partial functions $f$ and $g$,
we get a new partial function $f\circ g$,
whose domain is the set of strings
for which the computation of $f$ and of each copy of $g$
are all well-defined. Since partial functions are a generalization
of total Boolean functions, it is possible that the
composition conjecture holds for total functions but not
for partial functions. In this work, we will mainly focus
on the more general partial function setting;
when we do not mention anything about $f$ or $g$, they
should be assumed to be partial Boolean functions.

\subsection{Previous work}

\iffocs
\else
  \paragraph{Direct sum and product theorems.}
\fi
Direct sum theorems and direct product theorems
study the complexity of $\Id\circ g$,
where $g$ is an arbitrary Boolean function
but $\Id\colon\B^n\to\B^n$ is the identity function.
These are not directly comparable to composition theorems,
but they are of a similar flavor.

Jain, Klauck, and Santha~\cite{JKS10} showed that randomized query complexity
satisfies a direct sum theorem. Drucker~\cite{Dru12} showed
that randomized query complexity also satisfies a
direct product theorem, which means that 
$\Id\circ g$ cannot be solved too quickly
even with small success probability. More recently,
Blais and Brody~\cite{BB19} proved a strong direct sum theorem, showing
that computing $n$ copies of $g$ can be even harder
for randomized query complexity than $n$ times the cost
of computing $g$ (due to the need for amplification).

\iffocs
\else
  \paragraph{Composition theorems for other complexity measures.}
\fi
Several composition theorems are known for measures
that lower bound $\R(f)$; as such, these theorems
can be used to lower bound $\R(f\circ g)$
in terms of some smaller measure of $f$ and $g$.

First, though it is not normally phrased this way, the composition
theorem for quantum query complexity~\cite{Rei11,LMR+11}
can be viewed
as a composition theorem for a measure which lower bounds
$\R(f)$, since $\Q(f)\le\R(f)$ for all $f$.
Interestingly, as a lower bound technique for $\R(f)$,
$\Q(f)$ turns out to be incomparable to the other
lower bounds on randomized query complexity for which
composition is known, meaning that this composition
theorem can sometimes be stronger than everything
we know how to do using classical techniques.

Tal~\cite{Tal13} and independently Gilmer, Saks,
and Srinivasan~\cite{GSS16} studied the composition
behavior of simple measures like sensitivity,
block sensitivity, and fractional block sensitivity.
The behavior turns out to be somewhat complicated,
but is reasonably well characterized in these works.

G\"{o}\"{o}s and Jayram \cite{GJ16} studied the composition
behavior of conical junta degree, also known as
approximate non-negative degree. This measure is
a powerful lower bound technique for randomized algorithms
and seems to be equal to $\R(f)$ for all but the most
artificial functions; however, G\"{o}\"{o}s and Jayram
were only able to prove a composition theorem for a variant
of conical junta degree, and the variant appears
to be weaker in some cases (or at least harder to use).

Ben-David and Kothari \cite{BK18} showed a composition theorem
for a measure they defined called randomized sabotage
complexity, denoted $\RS(f)$.
They showed that this measure is larger
than fractional block sensitivity, and incomparable
to quantum query complexity and conical junta degree.
It is also nearly quadratically related to $\R(f)$
for total functions.

\iffocs
\else
  \paragraph{Composition theorems with a loss in \texorpdfstring{$g$}{g}.}
\fi
There are also composition theorems are known that lower bound
$\R(f\circ g)$ in terms of $\R(f)$ and some smaller measure
of $g$. 

Ben{-}David and Kothari \cite{BK18}
also showed that $\R(f\circ g)=\Omega(\R(f)\RS(g))$,
for the randomized sabotage complexity measure $\RS(g)$ mentioned above.
Anshu et al.~\cite{AGJ+18} showed that
$\R(f\circ g)=\Omega(\R(f)\R_{1/2-n^{-4}}(g))$,
where $\R_{1/2-n^{-4}}(g)$ is the randomized query complexity
of $g$ to bias $n^{-4}$.
These two results can also be used to
give composition theorems of the form
$\R(f\circ h\circ g)=\Omega(\R(f)\R(h)\R(g))$,
where $f$ and $g$ are arbitrary Boolean functions
but $h$ is a fixed small gadget
designed to break up any ``collusion'' between
$f$ and $g$. \cite{BK18}
proved such a theorem when $h$ is the index
function, while \cite{AGJ+18} proved it when
$h$ is the parity function of size $O(\log n)$.

Finally, Gavinsky, Lee, Santha, and Sanyal~\cite{GLSS19}
showed that $\R(f\circ g)=\Omega(\R(f)\bar{\chi}(g))$,
where $\bar{\chi}(g)$ is a measure they define.
They showed that $\bar{\chi}(g)=\Omega(\RS(g))$
and that $\bar{\chi}(g)=\Omega(\sqrt{\R(g)})$
(even for partial functions $g$), which
means their theorem also shows
$\R(f\circ g)=\Omega(\R(f)\sqrt{\R(g)})$.

\iffocs
\else
  \paragraph{Composition theorems with a loss in \texorpdfstring{$f$}{f}.}
\fi
There have been very few composition theorems
of the form $\R(f\circ g)=\Omega(M(f)\R(g))$ for some
measure $M(f)$. 
G\"{o}\"{o}s, Jayram, Pitassi, and Watson~\cite{GJPW18} showed that
$\R(\AND_n\circ g)=\Omega(n\R(g))$, which can be generalized
to $\R(f\circ g)=\Omega(\s(f)\R(g))$, where $\s(f)$
denotes the sensitivity of $f$.

Extremely recently, in work concurrent with this one,
Bassilakis, Drucker, G\"{o}\"{o}s, Hu, Ma, and Tan
\cite{BDG+20} showed that $\R(f\circ g)=\Omega(\fbs(f)\R(g))$, 
where $\fbs(f)$ is the fractional block sensitivity of $f$.
(This result also follows from our independent work in this paper.)

\iffocs
\else
  \paragraph{A relational counterexample to composition.}
\fi
Gavinsky, Lee, Santha, and Sanyal~\cite{GLSS19}
showed that the randomized composition
conjecture is false when $f$ is allowed to
be a \emph{relation}.
Relations are generalizations of partial functions,
in which $f$ has non-Boolean output alphabet and there
can be multiple allowed outputs for each input string.
The authors exhibited a family of relations $f_n$
and a family of partial functions $g_n$
such that $\R(f_n)=\Theta(\sqrt{n})$, $\R(g_n)=\Theta(n)$,
but $\R(f_n\circ g_n)=\Theta(n) \ll n^{3/2}$.

This counterexample of Gavinsky, Lee, Santha, and Sanyal
does not directly answer the randomized composition conjecture
(which usually refers to Boolean functions only),
but it does place restrictions on the types of tools
which might prove it true, since it appears that most or all
of the
composition theorems mentioned above do not use the fact
that $f$ has Boolean outputs and apply equally well
when $f$ is a relation---meaning those techniques
cannot be used to prove the composition conjecture true
without major new ideas.

\subsection{Our results}

Our first result shows that the randomized
composition conjecture is false for partial functions.

\begin{theorem}\label{thm:counterexample}
There is a family of partial Boolean functions $f_n$
and a family of partial Boolean functions $g_n$ such that
$\R(f_n)\to\infty$ and $\R(g_n)\to\infty$ as $n\to\infty$, but
\[\R(f_n\circ g_n)=O\left(\R(f_n)^{2/3}\R(g_n)^{2/3}\log^{2/3}\R(f_n)\right).\]
\end{theorem}

In this counterexample, $\R(f\circ g)$ is
polynomially smaller than what it was conjectured to be in the randomized composition conjecture.
However, this counterexample actually
uses functions $f$ and $g$ for which $\R(f)$ and
$\R(g)$ are logarithmic in the input
size of $f$. Therefore, the following slight weakening of the original randomized composition conjecture is still
viable.

\begin{conjecture}\label{conj:composition_log}
For all partial Boolean functions $f$ and $g$,
\[\R(f\circ g)=\Omega\left(\frac{\R(f)\R(g)}{\log n}\right),\]
where $n$ is the input size of $f$.
\end{conjecture}

Hence, even for partial functions, the composition story
is far from complete. This is in contrast to the setting in which
$f$ is a relation, where in the counterexample of~\cite{GLSS19}, 
the query complexity $\R(f\circ g)$ is
smaller than $\R(f)\R(g)$ by a polynomial factor
even relative to the input size.

\bigskip
Our second contribution is a new composition theorem
for randomized algorithms with a loss only in terms of $f$.

\begin{theorem}\label{thm:composition}
For all partial functions $f$ and $g$,
\[\R(f\circ g)=\Omega(\noisyR(f)\R(g)).\]
\end{theorem}

Here $\noisyR(f)$ is a measure we introduce, which is defined
as the cost of computing $f$ when given noisy oracle
access to the input bits; for a full definition,
see
\iffocs
\else
  \defn{noisyR}. 
\fi
As it turns out, $\noisyR(f)$
has a very natural interpretation, as the following theorem
shows.

\begin{theorem}\label{thm:noisyR_gapmaj}
For all partial functions $f$, we have
\[\noisyR(f)=\Theta\left(\frac{\R(f\circ\gapmaj_n)}{n}\right),\]
where $n$ is the input size of $f$ and $\gapmaj_n$
is the majority function on $n$ bits with the promise that the
Hamming weight of the input is either
$\lceil\frac{n}{2}+\sqrt{n}\rceil$
or $\lfloor\frac{n}{2}-\sqrt{n}\rfloor$.
Note that $\R(\gapmaj_n)=\Theta(n)$.
\end{theorem}

In other words, $\noisyR(f)$ characterizes the cost of
computing $f$ when the inputs to $f$ are given as $\sqrt{n}$-gap
majority instances
(divided by $n$, so that $\noisyR(f)\le\R(f)$).
This means that our composition theorem reduces the randomized
composition problem on arbitrary $f$ and $g$ to the randomized
composition problem of $f$ with $\gapmaj_n$.

\begin{corollary}
For all partial functions $f$ and $g$, we have
\[\R(f\circ g)=\Omega\left(
\frac{\R(f\circ \gapmaj_n)}{\R(\gapmaj_n)}\cdot\R(g)\right),\]
where $n$ is the input size of $f$.
\end{corollary}

These results hold even when $f$ is a relation.
We also note that the counterexamples to composition theorems---the 
one for partial functions in \thm{counterexample} and the relational
one in~\cite{GLSS19}---use the same function
$\gapmaj$ as the inner function $g$ (or close variants of it).
Therefore, there is a strong sense in which
$g=\gapmaj$ function is the only interesting case
for studying the randomized composition behavior of
$\R(f\circ g)$.

Next, we observe that our composition theorem is the strongest
possible theorem of the form $\R(f\circ g)=\Omega(M(f)\R(g))$
for any complexity measure $M$ of $f$. Formally, we have the following.

\begin{lemma}\label{lem:strongest}
Let $M(\cdot)$ be any positive-real-valued measure
of Boolean functions. Suppose that for all (possibly partial)
Boolean functions $f$ and $g$, we have
$\R(f\circ g)=\Omega(M(f)\R(g))$.
Then for all $f$, we have $\noisyR(f)=\Omega(M(f))$.
\end{lemma}

\begin{proof}
By \thm{noisyR_gapmaj}, we have
\[n\cdot\noisyR(f)=\Omega\big(\R(f\circ\gapmaj_n)\big),\]
where $n$ in the input size of $f$.
Now, by our assumption on $M(\cdot)$, taking $g = \gapmaj_n$ we obtain
{
\[\R(f\circ\gapmaj_n)=\Omega\big(M(f)\R(\gapmaj_n)\big)
=\Omega\big(M(f)\cdot n\big).\]
}
Hence $\noisyR(f)=\Omega(M(f))$, as desired.
\end{proof}

The natural next step is to study the measure
$\noisyR(f)=\R(f\circ\gapmaj_n)/n$.
We observe in 
\iffocs
\else
  \lem{noisyR_fbs}
\fi
that $\noisyR(f)=\Omega(\fbs(f))$.
However, we believe that a much
stronger lower bound should be possible.
The following conjecture is equivalent to
\conj{composition_log}.

\begin{conjecture}[Equivalent to \conj{composition_log}]
\label{conj:noisyR_log}
For all (possibly partial) Boolean functions $f$,
\[\noisyR(f)=\Omega\left(\frac{\R(f)}{\log n}\right).\]
\end{conjecture}

The equivalence of the two conjectures follows from
\thm{composition} in one direction, and from
\lem{strongest} in the other direction
(taking $M(f)= \R(f)/\log n$).

One major barrier for proving
\conj{noisyR_log} is that it is false for relations.
Indeed, the family of relations $f$ from \cite{GLSS19}
has $\noisyR(f)=O(1)$ and $\R(f)=\Omega(\sqrt{n})$.
Any lower bound $M(\cdot)$ for $\noisyR(\cdot)$
must therefore either be specific to functions
(and not work for relations), or else must
satisfy $M(f)=O(1)$ for that family of relations,
even though $\R(f)=\Omega(\sqrt{n})$
(which means $M(f)$ is a poor
lower bound on $\R(f)$, at least for some relations).

We are able to overcome this ``relational barrier''
for proving $\noisyR(f)$ lower bounds in the setting
of non-adaptive algorithms. Let
$\R^\na(f)$ denote the non-adaptive randomized
query complexity of $f$ and let
$\noisyR^\na(f)$ denote the non-adaptive version
of $\noisyR(f)$. Then for the family of relations $f$
from \cite{GLSS19}, it is still the case that
$\noisyR^\na(f)=O(1)$ and $\R^\na(f)=\Omega(\sqrt{n})$.
Despite this relational barrier, we have the following
theorem for the non-adaptive setting.

\begin{theorem}\label{thm:nonadaptive}
For all (possibly partial) Boolean functions $f$,
we have $\noisyR^\na(f)=\Theta(\R^\na(f))$.
\end{theorem}

Since \thm{nonadaptive} is false for relations,
its proof necessarily ``notices'' whether $f$
is a relation or a partial function. Such proofs
are unusual in query complexity. We hope that the
techniques we used in the proof of \thm{nonadaptive}
will assist future work in settling
the relationship between $\noisyR(f)$ and $\R(f)$
(perhaps resolving \conj{noisyR_log}).

\subsection{Our techniques}

\subsubsection{Main idea for the counterexample}

The main idea for the counterexample to composition
is to take $g=\gapmaj_m$ and to construct a function
$f$ that only requires some of its bits to be computed
to bias $1/\sqrt{m}$ instead of exactly. Achieving
bias $1/\sqrt{m}$ will be disproportionately cheap
for an input to $f\circ g$ compared to an input to $f$.

This is the same principle used for the relational
counterexample of \cite{GLSS19}. There, the authors
took $f$ to be the relational problem of taking
an input $x\in\B^n$ and returning an output
$y\in\B^n$ with the property that $|x-y|\le n/2-\sqrt{n}$.
This can be done using either $\sqrt{n}$
exact queries to $x$, or using $n$ queries to $x$
with bias $1/\sqrt{n}$ each. When $f$ is composed with
$g$ and $n=m$, it's not hard to verify that
$\R(f\circ g)=O(n)$, even though $\R(f)=\Omega(\sqrt{n})$
and $\R(g)=\Omega(m)=\Omega(n)$.

To convert $f$ into a partial Boolean function,
we use the indexing trick. We let the first $m$
bits of $f$ represent a string $x$, and we want to
force an algorithm to find a string $y$ that's within
Hamming weight $m/2-\sqrt{m}$ of $x$. To do so,
we can try adding an array of length $2^m$ to the
input of $f$, with entries indexed by $y$.
We'll fill the array with $*$ on positions indexed
by strings $y$ that are far from $x$. On positions
corresponding to strings $y$ within $m/2-\sqrt{m}$
of $x$, we'll put either all $0$s or all $1$s,
and we'll require the algorithm to output $0$
in the former case and $1$ in the latter case
(promised one of the two cases hold).

The above construction doesn't quite work,
because a randomized algorithm can cheat:
instead of finding a string $y$ close to $x$,
it can simply search the array for a non-$*$
bit and output that bit. Since a constant fraction
of the Boolean hypercube is within $m/2-\sqrt{m}$
of $x$, this strategy will succeed after a constant
number of queries. To fix this, all we need
to do is increase the gap from $\sqrt{m}$
to $10\sqrt{m\log m}$, so that $y$
is required to be
within $m/2-10\sqrt{m\log m}$ of $x$.
Now the non-$*$ positions in the array will fill only
a $1/m^{\Omega(1)}$ fraction of the array, and a randomized
algorithm has no hope of finding one of those positions with a small number of random guesses.
The input size of $f$ will be $n=m+2^m$.
Then we have
$\R(f)=\Theta(\sqrt{m\log m}))$,
$\R(g)=\Theta(m)$, but
$\R(f\circ g)=\Theta(m\log m)$
as we can solve $f\circ g$ by querying
each of the first $m$ copies of $g$
$O(\log m)$ times each, getting bias
$\Omega(\sqrt{(\log m)/m})$ for each of the $m$
bits of $x$, which provides a good string $y$
with high probability.

\subsubsection{Main idea for the composition theorem}

The main idea for proving the composition theorem
$\R(f\circ g)=\Omega(\noisyR(f)\R(g))$ is to try
to turn an algorithm for $f\circ g$ into an algorithm
for $f$. This is the standard approach for most
composition theorems, and the main question becomes
how to solve $f$ when we only have an algorithm $A$
which makes queries to an $nm$-length input for $f\circ g$.
When the algorithm
queries bit $j$ inside copy $i$ of $g$,
and we only have an $n$-bit input $x$ to $f$, what do we
query?

One solution would be to fix hard distributions
$\mu_0$ and $\mu_1$ for $g$, and then, when
$A$ makes a query to bit $j$ inside copy $i$ of $g$,
we can query $x_i$, sample an $m$-bit string
from $\mu_{x_i}$, and then return the $j$-th bit
of that string. However, this uses a lot of queries:
in the worst case, one query to $x$ would be needed
for each query $A$ makes, giving only the
upper bound $\R(f)\le\R(f\circ g)$ instead
of something closer to $\R(f)\le\R(f\circ g)/\R(g)$.
The goal is to simulate the behavior of $A$
while avoiding making queries to $x$
as much as possible.

One insight (also used in previous work)
is that if bit $j$ is queried inside copy $i$
of $g$, we only need to query $x_i$ from the real
input $x$ if $\mu_0$ and $\mu_1$ disagree
on the $j$-th bit with substantial probability.
In \cite{GLSS19}, the approach was to
first try to generate the answer $j$
from $\mu_0$ and $\mu_1$, and see if they happen
to agree; this way, querying the real input $x_i$ is only 
needed in case they disagree.

We do something slightly different: we assume
we have access to a (very) noisy oracle for $x_i$, and use
calls to the oracle to generate
bit $j$ from $\mu_{x_i}$ without actually finding
out $x_i$. In effect, this lets us use
the squared-Hellinger
distance between the marginal distributions
$\mu_0|_j$ and $\mu_1|_j$ as the cost
of generating the sample, instead of using
the total variation distance between $\mu_0|_j$ and
$\mu_1|_j$. That is, we charge a cost for the noisy oracle calls
in a special way, which ensures that the total
cost of the noisy oracle calls will be proportional
to the squared-Hellinger distance between
the transcript of $A$ when run on $\mu_0$
and when run on $\mu_1$. In other words,
the cost our $\R(f)$ algorithm pays for simulating
$A$ will be proportional to how much $A$
solved the copies of $g$, as tracked by the
Hellinger distance of the transcript of $A$
(i.e.\ its set of queries and query answers)
on $\mu_0$ vs.\ $\mu_1$.
It turns out this way of tracking the progress
of $A$ in solving $g$ is tight, at least
for the appropriate choice of hard distributions
$\mu_0$ and $\mu_1$ for $g$.
Therefore, this will give us an algorithm
for $f$ that has only $\R(f\circ g)/\R(g)$
cost, though this algorithm for $f$ will require
noisy oracles for the bits of the input---that is to say, 
it will be a $\noisyR(f)$
algorithm instead of an $\R(f)$ algorithm.

One wrinkle is that the hard distribution
produced by Yao's minimax theorem is not sufficient
to give the hardness guarantee we will need
from $\mu_0$ and $\mu_1$. Roughly speaking,
we will need $\mu_0$ and $\mu_1$ to be such that
distinguishing them with squared-Hellinger distance
$\epsilon$ requires at least $\Omega(\epsilon\R(g))$
queries, uniformly across  all choices of $\epsilon$.
To get such a hard distribution,
we use our companion paper \cite{BB20a}.
The concurrent work of \cite{BDG+20} also
gives a sufficiently strong hard distribution
for $g$ (though it is phrased somewhat differently).

\subsubsection{Noisy oracle model}

The noisy oracle model we will use is the following.
There is a hidden bit $b\in\B$ known to the oracle.
The oracle will accept queries with any parameter
$\gamma\in[0,1]$, and will return a bit $\tilde{b}$
that has bias $\gamma$ towards $b$---that is, 
a bit from $\Bernoulli\big(\frac{1-(-1)^b\gamma}2\big)$
(independently sampled for each query call).
This oracle can be called any number of times
with possibly different parameters,
but each call with parameter $\gamma$ costs $\gamma^2$.
(The cost $\gamma^2$ is a natural choice,
as it would take $O(1/\gamma^2)$ bits of bias
$\gamma$ to determine the bit with constant error.)

The measure $\noisyR(f)$ is defined as the cost
of computing $f$ (to worst-case bounded error) using
noisy oracle access to each bit in the input of $x$.
That is, instead of receiving query access to the
$n$-bit string $x$, we now have access to $n$ noisy
oracles, one for each bit $x_i$ of $x$.
We can call each oracle with any parameter $\gamma$
of our choice, at the cost of $\gamma^2$ per such call.
The goal is to compute $f$ to bounded error
using minimum expected cost (measured in the worst case
over inputs $x$). We note that by using
$\gamma=1$ each time, this reverts to the usual
query complexity of $f$, meaning that
$\noisyR(f)\le\R(f)$.

The key to our composition theorem
lies in using such a noisy oracle
for a bit $x_i$ to generate a sample from a distribution
$\mu_{x_i}|_j$ (distribution $\mu_{x_i}$ marginalized
to bit $j$) without learning $x_i$. More generally, suppose
we have two distributions, $p_0$ and $p_1$,
and we wish to sample from one of them, but we don't
know which one. The choice of which distribution
to sample from depends on a hidden bit $b$,
and we have noisy oracle access to $b$.
Suppose we know that $p_0$ and $p_1$
are close, say $\h^2(p_0,p_1)=\epsilon$.
How many queries to this noisy oracle do we need
to make in order to generate this sample?

We show that using such noisy oracle calls,
we can return a sample from $p_b$
with an expected cost of $O(\h^2(p_0,p_1))$.
When $p_0$ and $p_1$ are close, this is a much
lower cost than the $\Omega(1)$ cost of extracting $b$.
In other words, when the distributions are close,
we can return a sample from $p_b$ (without any error)
without learning the value of the bit $b$!
This is the key insight that allows our composition
result to work.

\subsubsection{Main idea for characterizing
\texorpdfstring{$\noisyR(f)$}{noisyR(f)}}

In order to show that
$\noisyR(f)=\Theta(\R(f\circ\gapmaj_n)/n)$,
we first note that the upper bound follows
from our composition theorem: that is,
$\R(f\circ\gapmaj_n)=\Omega(\noisyR(f)\R(\gapmaj_n))$,
and $\R(\gapmaj_n)=\Theta(n)$.
For the lower bound direction, we need
to convert a $\noisyR(f)$ algorithm
(which makes noisy oracle calls to the input bits,
with cost $\gamma^2$ for a noisy oracle call with parameter
$\gamma$) into an algorithm for
$\noisyR(f\circ\gapmaj_n)$ where each query
costs $1/n$. Recalling that $\gapmaj_n$
is the majority function with the promise that the
Hamming weight of the input is $n/2\pm\lfloor\sqrt{n}\rfloor$,
it's not hard to see that a single random query
to a $\gapmaj_n$ gadget (with cost $1/n$ each)
is the same thing as a noisy oracle query
with $\gamma\approx1/\sqrt{n}$.
Also, querying all $n$ bits in a $\gapmaj_n$
(with cost $1$ in total) is the same thing as
a noisy oracle query with $\gamma=1$.

To finish the argument, all we have to show
is that a $\noisyR(f)$ algorithm can always
be assumed to make only queries with $\gamma=1/\sqrt{n}$
or $\gamma=1$.
Now, it is well-known that an oracle with bias $\gamma$
can be amplified to an oracle with bias
$\gamma'>\gamma$ by calling it $O(\gamma'^2/\gamma^2)$
times and taking the majority of the answers.
Since oracle calls with parameter $\gamma$
cost us $\gamma^2$, this fact ensures
that we only need to make noisy oracle
calls with parameter either $\gamma=\hat{\gamma}$
or $\gamma=1$, where $\hat{\gamma}$ is
extremely small -- smaller than anything
used by an optimal (or at least near-optiomal)
$\noisyR(f)$ algorithm. This is because for any desired
bias level larger than $\hat{\gamma}$, we could
simply amplify the $\hat{\gamma}$ calls.

Hence it only remains to show how to simulate
noisy oracle queries with an arbitrarily
small parameter $\hat{\gamma}$ using
noisy oracle queries with parameter $1/\sqrt{n}$.
For this, we consider a random walk on a line
that starts at $0$ and flips a
$\Bernoulli\big(\frac{1-(-1)^b\hat{\gamma}}{2}\big)$ coin
when deciding whether step forwards or backwards.
Consider making this walk starting at $0$,
walking until either $k$ or $-k$ is reached,
and then stopping (where $k$ is some fixed integer).
Note that the probability that neither $k$ or $-k$
is ever reached after infinitely many steps is $0$.
We then make the following key observation:
the probability distribution over
the sequence steps of this walk,
\emph{conditioned} on reaching $k$ before $-k$,
is the same whether $b=0$ or $b=1$.
Therefore, it is possible to generate
the full walk by generating the
sequence of multiples of $k$ the walk will reach
(in a way that depends on $b$), and then completely
separately -- and independently of $b$ -- generating
the sequence of steps between one multiple and the next,
up to negation.

To simulate a bias $\hat{\gamma}$ oracle with a
bias $1/\sqrt{n}$ oracle, we can use latter
to generate the sequence of multiples of $k$ described
above, with $k=O(1/(\sqrt{n}\hat{\gamma}))$. We generate
this sequence one at a time. For each one, we can then
generate $\ell$ calls to the bias $\hat{\gamma}$
oracle, where $\ell$ is the (random) number of steps the random
walk takes to go from one multiple of $k$ to the next.
This simulation is perfect: is produces the distribution
of any number of calls to the $\hat{\gamma}$-bias oracle.
It also turns out to use the right number
of noisy oracle queries in the long run. The only
catch is that if the algorithm makes only
one noisy oracle call with bias $\hat{\gamma}$,
this still requires one call to the oracle of bias $1/\sqrt{n}$,
at a cost of $1/n$ instead of $1/\hat{\gamma}^2$.
Since there are $n$ total bits, this means the simulation
can suffer an additive cost of $1$.
To complete the argument, we then
show that $\noisyR(f)=\Omega(1)$
for every non-constant Boolean function $f$.

\subsubsection{Main idea for bypassing the relational
barrier in the non-adaptive setting}

The trick for showing $\noisyR^\na(f)=\Theta(\R^\na(f))$
for partial functions is to use an information-theoretic
characterization of this statement. First,
using a Yao-style minimax theorem, we can assume we are working
against a hard distribution $\mu$ for $\R^\na(f)$.
Then we consider a non-adaptive randomized algorithm
that uses noisy oracle queries (that is, a $\noisyR^\na(f)$
algorithm) that solves $f$ against $\mu$. By some simple
modifications and reductions, we can assume that this algorithm
simply makes one noisy query to each bit of the input,
with bias parameter $1/\sqrt{n}$. In other words, if $X$
is the random variable for a string sampled from $\mu$,
and if $Y$ is the random variable we get by flipping
each bit of $X$ independently with probability $(1-1/\sqrt{n})/2$,
then we can assume a $\noisyR^\na(f)$ algorithm just has
access to the string $Y$ and tries to compute $f(X)$ using $Y$.
Our reductions change the length of the string (by duplicating
bits of the input), and the
cost of this noisy randomized algorithm will roughly be
$|X|/n$, where $|X|$ is the length of the string $X$
and $n$ is the length of the original string.

What we wish to show is that such a noisy non-adaptive
randomized algorithm
(which computes $f(X)$ using $Y$) can be converted into a
regular non-adaptive randomized algorithm which
computes $f(X)$ by querying
only around $|X|/n$ bits of $X$. To do so, we use
a theorem of Samorodnitsky \cite{Sam16,PW17}, which states
that the erasure channel with parameter $\rho^2$ -- which
deletes each bit of $X$ with probability $1-\rho^2$ --
preserves more information about any function $f(X)$
than the noisy channel
with parameter $\rho$ (which flips each bit of $X$ with probability
$(1-\rho)/2$). Hence, if $f(X)$ can be computed from $Y$,
it can also be computed from the string $Z$ which is formed
by deleting each bit of $X$ with probability $1-1/n$. Since
$Z$ reveals only $|X|/n$ bits on expectation, this can be used
to define a non-adaptive randomized algorithm whose cost
is at most $\noisyR^\na(f)$, and which still succeeds in computing
$f$ against $\mu$ to bounded error.
This shows $\R^\na(f)=O(\noisyR^\na(f))$.

We note that the step where we used the fact that
$f$ is a partial function is the step where we said that if
$Z$ gives \emph{information} about $f(X)$, seeing $Z$
can be used to \emph{compute} $f(X)$ to bounded error.
This statement holds when $f(X)$ is a Boolean-valued random variable,
but it has no good analogue in the relational setting
(and indeed, we know that $\noisyR^\na(f)$ does not equal
$\R^\na(f)$ for relations).

\begin{fulltext}

\section{Preliminaries and definitions}

\subsection{Query complexity}

We introduce some basic concepts in query complexity.
For a survey, see \cite{BdW02}. Fractional
block sensitivity can be found in
\cite{Aar08,KT16}.

\paragraph{Partial Boolean functions.}
In this work, we will refer to partial
Boolean functions, which are functions
$f\colon S\to\B$ where $S\subseteq\B^n$
and $n$ is a positive integer. For a partial
function $f$, the term \emph{promise}
refers to its domain $S$, which we also
denote by $\Dom(f)$. If $\Dom(f)=\B^n$,
we say $f$ is a total function.

\paragraph{Composition.}
For partial Boolean functions $f$ and $g$
on $n$ and $m$ bits respectively, we define
their \emph{composition}, denoted $f\circ g$,
as the Boolean function on $nm$ bits with
the following properties.
$\Dom(f\circ g)$ will contain the set of
$nm$-bit strings which are concatenations
of $n$ different $m$-bit strings in $\Dom(g)$,
say $x^1,x^2,\dots,x^n$, where the tuple
$(x^1,x^2,\dots,x^n)$ must have the property
that the string $g(x^1)g(x^2)\dots g(x^n)$
is in $\Dom(f)$. The value of $f\circ g$
on such a string $x^1x^2\dots x^n$
is then defined as $f(g(x^1)g(x^2)\dots g(x^n))$.

\paragraph{Partial assignments.}
A partial assignment is a string in $\{0,1,*\}^n$
representing partial knowledge of a string in $\B^n$.
We say two partial assignments $w$ and $z$ are consistent
if they agree on the non-$*$ bits, that is,
for every $i\in[n]$ we have
either $w_i=*$ or $z_i=*$ or $w_i=z_i$
(we use $[n]$ to denote $\{1,2,\dots,n\}$).

\paragraph{Decision trees.}
A decision tree $D$ on $n$ bits
is a rooted binary tree whose leaves are labeled
by $\B$ and whose internal nodes are labeled by $[n]$.
We do not allow two internal nodes of a decision
tree to have the same label if one is a descendant
of the other. We interpret a decision tree $D$
as a deterministic algorithm which takes
as input a string $x$, starts at the root,
and at each internal node with label $i$,
the algorithm queries $x_i$ and then
goes left down the tree if $x_i=0$ and right
if $x_i=1$. When this algorithm reaches a leaf,
it outputs its label. We denote by
$\cost(D,x)$ the number of queries $D$
makes when run on $x$, and by $\cost(D)$
the height of the tree $D$. We denote the
output of $D$ on input $x$ by $D(x)$.
We say $D$ computes Boolean function $f$
if $D(x)=f(x)$ for all $x\in\Dom(f)$.

\paragraph{Randomized decision trees.}
A randomized decision tree $R$ on $n$ bits
is a probability distribution over deterministic
decision trees on $n$ bits.
We denote by $\cost(R,x)$ the expectation
of $\cost(D,x)$ over decision trees $D$
sampled from $R$. If $\mu$ is a distribution
over $\B^n$, we further denote by $\cost(R,\mu)$
the expectation of $\cost(R,x)$ over $x$
sampled from $\mu$. We denote by $\cost(R)$
the maximum of $\cost(R,x)$ over $x\in\B^n$,
and by $\height(R)$ the maximum of $\cost(D)$
over $D$ in the support of $R$.
Further, we let $R(x)$ denote
the random variable $D(x)$ with $D$ sampled from $R$.
We say $R$ computes $f$ to error $\epsilon\in[0,1/2]$
if $\Pr[R(x)=f(x)]\ge 1-\epsilon$ for all
$x\in\Dom(f)$.

\paragraph{Randomized query complexity.}
The randomized query complexity of a Boolean function
$f$ to error $\epsilon$,
denoted $\R_\epsilon(f)$, is the minimum height
$\height(R)$ of a randomized decision tree computing
$f$ to error $\epsilon$. The expectation
version of the randomized query complexity of $f$,
denoted $\bar{\R}_\epsilon(f)$, is the minimum
value of $\cost(R)$ of a randomized decision tree
computing $f$ to error $\epsilon$.
When $\epsilon=1/3$, we omit it and write
$\R(f)$ and $\bar{\R}(f)$. We note that
randomized query complexity can be amplified
by repeating the algorithm a few times and taking
the majority vote of the answers; for this reason,
the constant $1/3$ is arbitrary and any other
constant in $(0,1/2)$ could work for the definition.
Note that in the constant error regime,
$\bar{R}(f)=\Theta(\R(f))$, since we can cut off
paths of a $\bar{R}(f)$ algorithm
that run too long and use Markov's inequality
to argue that we only suffer a constant error
penalty for this.

\paragraph{Block sensitivity.}
Let $f$ be a Boolean function and let $x\in\Dom(f)$.
A \emph{sensitive block} of $f$ at $x$
is a subset $B\subseteq[n]$ such that
$x^B\in\Dom(f)$ and $f(x^B)\ne f(x)$,
where $x^B$ denotes the string $x$ with bits
in $B$ flipped (i.e.\ $x^B_i=x_i$ for $i\notin B$
and $x^B_i=1-x_i$ for $i\in B$).
The \emph{block sensitivity} of $f$
at $x$, denoted $\bs_x(f)$,
is the maximum number of disjoint sensitive
blocks of $f$ at $x$. The block sensitivity
of $f$, denoted $\bs(f)$, is the maximum value
of $\bs_x(f)$ over $x\in\Dom(f)$.
We note that $\R(f)=\Omega(\bs(f))$,
since if $B_1,\dots,B_k$ are disjoint sensitive
blocks of $f$ at $x$, then a randomized
algorithm must make $\Omega(k)$ queries
to determine whether the input is $x$ or $x^{B_j}$
for some $j\in[k]$.

\paragraph{Fractional block sensitivity.}
Fix a Boolean function $f$ and an input $x\in\Dom(f)$,
and let $\mathcal{B}$ be the set of all sensitive
blocks of $f$ at $x$. We consider weighting schemes
assigning non-negative weights $w_B$ to blocks
$B\in\mathcal{B}$. We say such a scheme is feasible
if for each $i\in[n]$, the sum of $w_B$
over all blocks $B\in\mathcal{B}$ containing $i$
is at most $1$. The fractional block sensitivity
of $f$ at $x$, denoted $\fbs_x(f)$, is the maximum
total weight in such a feasible weighting scheme.
The fractional block sensitivity of $f$,
denoted $\fbs(f)$, is the maximum of $\fbs_x(f)$
over all $x\in\Dom(f)$.
We note that $\R(f)=\Omega(\fbs(f))$.
To see this, let $R$ be a randomized algorithm
solving $f$ let $x\in\Dom(f)$ be an input,
and for $i\in[n]$ let $p_i$
be the probability that $R$ queries bit $i$
when run on $x$. If, for any sensitive block $B$,
we have $\sum_{i\in B} p_i\ll 1$, then
$R$ does not distinguish $x$ from $x^B$ with
constant probability, which means $R$
fails to compute $f$ to bounded error
(since $f(x)\ne f(x^B)$). So we have
$\sum_{i\in B} p_i\ge\Omega(1)$ for all $B$.
Then 
\[\height(R)\ge \sum_{i\in[n]} p_i
\ge \sum_{i\in[n]}p_i\sum_{B\in\mathcal{B}:i\in B} w_B
=\sum_{B\in\mathcal{B}} w_B\sum_{i\in B} p_i
=\Omega\left(\sum_{B\in\mathcal{B}} w_B\right)
=\Omega(\fbs_x(f)).\]

\paragraph{Relations.}
A relation $f$ is a subset of $\B^n\times\Sigma$
for some finite alphabet $\Sigma$. When computing
a relation $f$, we only require that an algorithm
$A$ given input $x$
outputs some $\sigma\in\Sigma$ satisfying
$(x,\sigma)\in f$. In other words, each input
may have many valid outputs. It is not hard
to generalize the definitions of $\D(f)$
and $\R(f)$ to include relations: the decision
trees need leaves labeled by $\Sigma$, but otherwise
everything works the same
(though one catch is that amplification no longer works,
which means $\R_\epsilon(f)$ becomes a different measure
for different values of $\epsilon$).
Note that relations generalize partial functions,
because instead of restricting the inputs to a
promise set $S\subseteq\B^n$, we can simply allow
all possible outputs for every $x\notin S$.
With this in mind, it is not hard to see that
composition $f\circ g$ is well-defined if $f$ is a
relation, so long as $g$ remains a (possibly partial)
Boolean function. In general, we will define measures
for Boolean functions and later wish to apply
them to relations; this will usually work
without too much trouble.

\subsection{Distance measures for distributions}

In this work, we will only consider finite-support
distributions and finite-support random variables.
For a distribution $\mu$, we will use
$\mu^A$ to denote the conditional distribution of $\mu$
conditioned on event $A$. If $\mu$ is a distribution
over $\B^n$ and $z$ is a partial assignment,
we will also use $\mu^z$ to denote the distribution $\mu$
conditioned on the string sampled from $\mu$ agreeing with
the partial assignment $z$. If $\mu$ is a distribution
over $\B^n$ and $j\in[n]$ is an index, we will use
$\mu|_j$ to denote the marginal distribution of $\mu$
on the bit $j$ (the distribution
we get by sampling $x$ from $\mu$ and returning $x_j$).

The following distance measures will be useful.
All logarithms are base 2.

\begin{definition}[Distance measures]
\label{def:dist_measures}
For probability distributions $\mu_0$ and $\mu_1$
over a finite support $S$, define the squared-Hellinger,
symmetrized chi-squared, Jensen-Shannon, and total variation
distances respectively as follows:
\begin{align*}
\h^2(\mu_0,\mu_1) &\coloneqq \frac{1}{2}
    \sum_{x\in S}(\sqrt{\mu_0[x]}-\sqrt{\mu_1[x]})^2\\
\Ess^2(\mu_0,\mu_1) &\coloneqq \frac{1}{2}
    \sum_{x\in S}\frac{(\mu_0[x]-\mu_1[x])^2}{\mu_0[x]+\mu_1[x]}\\
\JS(\mu_0,\mu_1) &\coloneqq \frac{1}{2}
    \sum_{x\in S}\mu_0[x]\log\frac{2\mu_0[x]}{\mu_0[x]+\mu_1[x]}
    +\mu_1[x]\log\frac{2\mu_1[x]}{\mu_0[x]+\mu_1[x]}\\
\Delta(\mu_0,\mu_1) 
  &\coloneqq \frac{1}{2} \sum_{x\in S}|\mu_0[x]-\mu_1[x]|.
\end{align*}
\end{definition}

We will need a few basic claims regarding the properties of various distance measures between probability distributions. The first one relates these
probability distributions to each other. This is
known in the literature, though the citations
are hard to trace down; some parts of this
inequality chain follow from \cite{Top00},
some parts from \cite{MCAL17}, and for others we cannot
find a good citation.
In any case, a proof of the complete chain is provided
in the appendix of our companion manuscript \cite{BB20a}.

\begin{claim}[Relationship of distance measures]
\label{clm:h_vs_JS_vs_S}
For probability distributions $\mu_0$ and $\mu_1$,
\[\h^2(\mu_0,\mu_1)\le \JS(\mu_0,\mu_1)
\le \Ess^2(\mu_0,\mu_1)\le 2\h^2(\mu_0,\mu_1).\]
We also have $\Delta^2(\mu_0,\mu_1)\le\Ess^2(\mu_0,\mu_1)\le\Delta(\mu_0,\mu_1)$.
\end{claim}

Since the distance measures
$\h^2$, $\Ess^2$, and $\JS$ are equivalent
up to constant factors, one might wonder why
we need all three. It turns out that the 
squared-Hellinger distance is mathematically
the nicest (e.g.\ it tensorizes and behaves
nicely under disjoint mixtures), the Jensen-Shannon
distance has an information-theoretic interpretation
that allows us to use tools from information theory,
and the symmetrized chi-squared distance $\Ess^2$
is the one that most naturally captures the cost
of outputting a sample from $\mu_b$
given noisy oracle access to the bit $b\in\B$
(see \lem{single_query_simulation}).

\subsubsection{Properties of the squared-Hellinger distance}

\begin{claim}[Hellinger tensorization]
\label{clm:tensorization}
Fix distributions $\mu_0$ and $\mu_1$ with finite support,
and let $\mu_0^{\otimes k}$ denote the distribution
where $k$ independent samples from $\mu_0$ are returned
(with $\mu_1^{\otimes k}$ defined similarly). Then
\[\h^2\left(\mu_0^{\otimes k},\mu_1^{\otimes k}\right)
=1-\left(1-\h^2(\mu_0,\mu_1)\right)^k.\]
\end{claim}

\begin{proof}
From the definition of $\h^2(\cdot,\cdot)$,
it is not hard to see that $\h^2(\mu_0,\mu_1)=1-F(\mu_0,\mu_1)$,
with $F(\mu_0,\mu_1)$ denoting
the fidelity $\sum_x\sqrt{\mu_0[x]\mu_1[x]}$ between
$\mu_0$ and $\mu_1$. The claim that
$F(\mu_0^{\otimes k},\mu_1^{\otimes k})=F(\mu_0,\mu_1)^k$
is easy to see, as it is simply the claim
\[\sum_{x_1}\sum_{x_2}\dots\sum_{x_k}\sqrt{\mu_0[x_1]\dots\mu_0[x_k]\cdot\mu_1[x_1]\dots\mu_1[x_k]}=\left(\sum_x\sqrt{\mu_0[x]\mu_1[x]}\right)^k. \qedhere\]
\end{proof}

\begin{claim}[Hellinger interpretation]
\label{clm:h_interpretation}
For distributions $\mu_0$ and $\mu_1$, let $k$ be the
minimum number of independent samples from $\mu_b$ necessary to be able
to deduce $b$ with error at most $1/3$. Then
\[k=\Theta\left(\frac{1}{\h^2(\mu_0,\mu_1)}\right),\]
with the constants in the big-$\Theta$ notation being universal.
\end{claim}

\begin{proof}
This minimum $k$ is the minimum $k$ such that
$\mu_0^{\otimes k}$ and $\mu_1^{\otimes k}$ can be distinguished
with constant error; it is well-known that this is the
same as saying $\Delta(\mu_0^{\otimes k},\mu_1^{\otimes k})$
is at least a constant. By \clm{h_vs_JS_vs_S},
this is the same as saying
$\h^2(\mu_0^{\otimes k},\mu_1^{\otimes k})$
is at least a constant. By \clm{tensorization},
this is the same as saying
$1-(1-\h^2(\mu_0,\mu_1))^k$ is at least a constant.
The function $1-(1-x)^k$ behaves like $kx$ when $k$
is small compared to $1/x$, so the minimum such $k$
must be $\Theta(1/\h^2(\mu_0,\mu_1))$.
\end{proof}

\begin{claim}[Hellinger of disjoint mixtures]
\label{clm:h_mix}
Let $p_a$ and $q_a$ be families of distributions, with $a$
ranging over a finite set $S$. Suppose that for each $a,b\in S$
with $a\ne b$, it holds that the support $U_a$ of
$p_a$ and $q_a$ is disjoint from the support $U_b$ of $p_b$
and $q_b$. Let $\mu$ be a distribution
over $S$. Let $p_\mu$ denote the distribution that samples
$a\leftarrow \mu$ and then returns a sample from $p_a$,
and let $q_\mu$ be defined similarly. Then
\[\h^2(p_\mu,q_\mu)=\bE_{a\sim\mu}[\h^2(p_a,q_a)].\]
\end{claim}

\begin{proof}
As in the proof of \clm{tensorization}, it suffices
to prove that the fidelity satisfies
$F(p_\mu,q_\mu)=\bE_{a\sim\mu}[F(p_a,q_a)]$.
This is clear, as it is simply the claim
\[\sum_{a\in S}\sum_{x\in U_a}\sqrt{\mu[a]p_a[x]\mu[a]q_a[x]}
=\sum_{a\in S}\mu[a]\sum_{x\in U_a}\sqrt{p_a[x]q_a[x]}. \qedhere\]
\end{proof}

\subsubsection{Properties of the Jensen-Shannon distance}

Here we will need some standard notation from information
theory. For random variables $X$ and $Y$ with finite supports,
we write $H(X)\coloneqq-\sum_x\Pr[X=x]\log\Pr[X=x]$ for the entropy
of $X$, and $I(X;Y)\coloneqq H(X)+H(Y)-H(X,Y)$
for the mutual information between $X$ and $Y$.
If $Z$ is another random variable, we will write
$I(X;Y|Z)\coloneqq \sum_z[\Pr(Z=z)\cdot I(X^{Z=z};Y^{Z=z})]$ for
the conditional mutual information, where we use
the notation $X^{Z=z}$ to denote the random variable $X$
conditioned on the event $Z=z$. We note that
$I(X;Y)=I(Y;X)$ and $I(X;Y|Z)=I(Y;X|Z)$.

The chain rule for mutual information is well-known.

\begin{claim}[Chain rule for mutual information]
\label{clm:chain_rule}
For discrete random variables $X$, $Y$, and $Z$, we have
\[I(X;Y|Z)=I(X,Z;Y)-I(Z;Y).\]
\end{claim}

We now use information theory
to characterize the Jensen-Shannon distance $\JS$.

\begin{claim}[Jensen-Shannon interpretation]
\label{clm:JS_I}
For finite-support probability distributions $\mu_0$ and $\mu_1$,
\[\JS(\mu_0,\mu_1)=I(X;\mu_X)\]
where $X$ is a $\Bernoulli(1/2)$ random variable.
\end{claim}

\begin{proof}
Let $\mu=(\mu_0+\mu_1)/2$.
We have
\[I(X;\mu_X)=H(X)+H(\mu_X)-H(X\mu_X)
=1+\sum_x \mu[x]\log\frac{1}{\mu[x]}
-\frac12\sum_x\mu_0[x]\log\frac2{\mu_0[x]}+\mu_1[x]\log\frac2{\mu_1[x]}\]
\[=1+\frac{1}{2}\sum_x\mu_0[x]\log\frac{\mu_0[x]}{\mu_0[x]+\mu_1[x]}
+\mu_1[x]\log\frac{\mu_1[x]}{\mu_0[x]+\mu_1[x]}.\]
This last line equals the definition of $\JS(\mu_0,\mu_1)$
by using $1=(1/2)\sum_x\mu_0[x]+\mu_1[x]$.
\end{proof}

We will also need to understand $I(Z;\mu_Z)$ when $Z$
is a Bernoulli distribution with parameter not quite
equal to $1/2$.

\begin{claim}[Information of imperfect coins]
\label{clm:JS_bias}
Let $Y_1$ and $Y_2$ be random variables drawn from distributions $\mu_0$ and $\mu_1$, respectively.
Let $X$ be a $\Bernoulli(1/2)$ random variable,
and let $Z$ be a $\Bernoulli((1+\gamma)/2)$ be a Bernoulli
random variable with bias $-1 \le \gamma \le 1$. 
Then
\[
I(Z;Y_Z) \ge (1-|\gamma|)I(X;Y_X)
         = (1-|\gamma|)\JS(p_0,p_1).
\]
\end{claim}

\begin{proof}
Consider the case where $\gamma \ge 0$.
Let $B \sim \Bernoulli(\gamma)$ and 
\[
Z \sim \begin{cases}
1 & \mbox{if } B = 1 \\
\Bernoulli(\tfrac12) & \mbox{otherwise.}
\end{cases}
\]
Then $Z \sim \Bernoulli(\tfrac{1+\gamma}{2})$. Using the fact that $B$ and $Y_Z$ are independent conditioned on $Z$, 
the chain rule, and the non-negativity of conditional mutual information, we obtain
\[
I(Z; Y_Z) = I(B,Z; Y_Z) = I(B; Y_Z) + I(Z; Y_Z \mid B)
\ge I(Z; Y_Z \mid B).
\]
Then
\[
I(Z; Y_Z \mid B) 
= (1-\gamma) I(Z; Y_Z \mid B=0) 
= (1-\gamma) I(X; Y_X).
\]
The case where $\gamma < 0$ is obtained by a symmetric argument.
\end{proof}

\subsection{Noisy oracles and the definition of
\texorpdfstring{$\noisyR(f)$}{noisyR(f)}}

We use the following sequence of definitions to define
$\noisyR(f)$.

\begin{definition}[Noisy oracles]\label{def:noisy_oracle}
A \emph{noisy oracle} to a bit $b \in \{0,1\}$ is an oracle that takes a parameter $\gamma$ in the range $-1 \le \gamma \le 1$ and outputs a random bit $a \in \{0,1\}$ that satisfies
$\Pr[ a = b ] = \frac{1+\gamma}{2}$.
We write $\textsc{NoisyOracle}_b(\gamma)$ to denote a call to the noisy oracle for bit $b$ with parameter $\gamma$.
Each call to a noisy oracle returns an independent random variable.
The \emph{cost} of a query to a noisy oracle with parameter $\gamma$ is defined to be $\gamma^2$.
\end{definition}

Note that the user of the noisy oracle is allowed
to \emph{choose} the bias parameter $\gamma$, and smaller
$\gamma$ comes with smaller cost.

\begin{definition}[Noisy oracle algorithms]
A \emph{noisy oracle decision tree} $D$ on $n$ bits
is a binary tree with internal nodes labeled by pairs
$(i,\gamma)$ with $i\in[n]$ and $\gamma\in[0,1]$,
and leaves labeled by $\B$.
Unlike for regular decision trees, we do not forbid
descendants from having the same label as ancestors.
We only allow finite decision trees.

A \emph{noisy oracle randomized algorithm} $R$ on $n$
bits is a finite-support probability distribution over
noisy oracle decision trees on $n$ bits. For $x\in\B^n$,
we let $R(x)$ be the random variable representing
the output of $R$ on $x$, defined as the result of
sampling a decision tree $D$ from $R$ and walking
down the tree to a leaf, where at each internal node labeled
$(i,\gamma)$ we call the noisy oracle for $x_i$
with parameter $\gamma$ and go to the left child if
the output is $0$ and to the right child if the output is $1$.
The cost of such a path to a leaf is the sum of $\gamma^2$
for parameters $\gamma$ in the path, and $\cost(R,x)$
denotes the expected cost of running $R$ on $x$.

We say that $R$ computes Boolean function $f$ to error $\epsilon$
if $\Pr[R(x)=f(x)]\ge 1-\epsilon$ for all $x\in\Dom(f)$.
\end{definition}

\begin{definition}[Noisy randomized query complexity]
\label{def:noisyR}
The \emph{$\epsilon$-error noisy randomized query complexity} of a (possibly partial) Boolean function $f$, denoted
$\noisyR_\epsilon(f)$,
is the infimum expected worst-case
cost of a noisy oracle randomized algorithm that computes $f$ to error $\epsilon$. In other words, the cost is measured
in the worst case against inputs $x\in\Dom(f)$,
but on expectation against the internal randomness
of the algorithm and against the randomness of the oracle answers.

When $\epsilon=1/3$, we omit it and write
We write $\noisyR(f)$.
\end{definition}

We note that the set of noisy oracle randomized algorithms
on $n$ bits is not compact, so the infimum in the definition
of $\noisyR_\epsilon(f)$ need not be attained. However,
this won't bother us too much, as there is always some
algorithm attaining (say) cost $2\noisyR_\epsilon(f)$
for computing $f$ to error $\epsilon$, and we will
not care about constant factors. We also note that
noisy oracle randomized algorithms can be amplified as usual,
which means that the constant $1/3$ is arbitrary.
Further, by cutting off paths that cost too much
and using Markov's inequality, it's not hard to see
that there is always an algorithm computing
$f$ to bounded error using noisy oracles whose worst-case cost
is $O(\noisyR(f))$ even in the absolute worst case
(getting maximally unlucky with oracle answers
and internal randomness).

The following well-known lemma will be very convenient
for analyzing low-bias oracles. For completeness,
we prove it in \app{amplify}.

\begin{restatable}[Small bias amplification]{lemma}{amplify}
\label{lem:amplify}
Let $\gamma\in[-1/3,1/3]$ be nonzero, and let $k$ be an odd positive
integer which is at most $1/\gamma^2$. Let $X$
be the Boolean-valued random variable we get by
generating $k$ independent bits from $\Bernoulli((1+\gamma)/2)$
and setting $X$ to their majority vote. Then
$X$ has distribution $\Bernoulli((1+\gamma')/2)$,
where $\gamma'\in[-1,1]$ has the same sign as $\gamma$ and
\[(1/3)\sqrt{k}|\gamma|\le |\gamma'|\le 3\sqrt{k}|\gamma|.\]
\end{restatable}

\subsection{Transcripts, Hellinger distinguishing cost,
and \texorpdfstring{$\sfR(g)$}{sfR(g)}}
\label{sec:shaltiel}

To get our composition theorem to work, we will need
to start with very hard $0$- and $1$-distributions for $g$.
We will prove our lower bound in a way that clarifies
the dependence on the hardness of these distributions:
the lower bound will be in terms of the \emph{Hellinger
distinguishing cost} of these distributions,
which we define below. We will then cite our
companion manuscript \cite{BB20a} to ensure
that there exist hard distributions for $g$ whose
Hellinger distinguishing cost is $\Omega(\R(g))$.

\begin{definition}[Transcript]\label{def:transcript}
Let $D$ be a decision tree
on $n$ bits, and let $x\in\B^n$.
The \emph{transcript} of $D$ when run on $x$, denoted
$\tran(D,x)$, is the sequence of pairs
$(i_1,x_{i_1}),(i_2,x_{i_2}),\dots,(i_T,x_{i_T})$
consisting of all queries $i_t\in[n]$ that $D$ makes and all
answers $x_{i_t}\in\B$ that $D$ receives to its queries, until a leaf
is reached.

The transcript of $D$ on a distribution $\mu$ of inputs is the
random variable which takes value $\tran(D,x)$ when $x$
is sampled from $\mu$.

Furthermore, if $R$ is a randomized decision tree and $\mu$
is a distribution over $\B^n$, we define the transcript of $R$
when run on $\mu$, denoted $\tran(R,\mu)$, to be the random
variable which evaluates to the pair $(D,\tran(D,x))$ when
$D$ is the decision tree sampled from $R$ and $x$ is the input
sampled from $\mu$. In other words, the transcript writes down
both the queries seen and the value of the
internal randomness of the algorithm.
\end{definition}

\begin{definition}[Hellinger distinguishing cost]
\label{def:Hellinger_distinguishing_cost}
Let $n\in\bN$ and let $\mu_0$ and $\mu_1$ be distributions
over $\B^n$. The \emph{Hellinger distinguishing cost} of
$\mu_0$ and $\mu_1$ is 
\[
\cost(\mu_0,\mu_1) \coloneqq \min_R
\frac{\min\{\cost(R,\mu_0),\cost(R,\mu_1)\}} 
     {\h^2(\tran(R,\mu_0),\tran(R,\mu_1))},
\]
where the minimum is taken over all randomized decision
trees $R$
and we interpret $x/0 = \infty$ for every $x \ge 0$
in the minimum.
\end{definition}

Informally, the Hellinger distinguishing cost measures
the number of queries a randomized algorithm must make
in order to ensure it behaves differently on $\mu_0$
and $\mu_1$. We allow algorithms to behave only a little
differently on $\mu_0$ and $\mu_1$ if their cost is low enough.

Next, we will define the ``Shaltiel free'' randomized
query complexity of $g$ as the maximum
Hellinger distinguishing cost between $0$- and $1$-distributions
of $g$. We name this measure $\sfR(g)$ after Shaltiel \cite{Sha03}
who showed that some distributions for a Boolean
function $g$ may be hard to compute to bounded
error without being sufficiently difficult
in other ways (e.g.\ they may be trivial to solve
to small bias).

\begin{definition}[Shaltiel-free randomized query complexity]
\label{def:sfR}
Let $g$ be a (possibly partial) function. The
\emph{Shaltiel-free randomized query complexity} of $g$, denoted
$\sfR(g)$, is the maximum over all distributions $\mu_0$ and
$\mu_1$ supported on $g^{-1}(0)$ and $g^{-1}(1)$, respectively,
of the Hellinger distinguishing cost of $\mu_0$ and $\mu_1$.
In other words,
\[
\sfR(g)\coloneqq\max_{\substack{\mu_0\,:\,\supp(\mu_0) \subseteq g^{-1}(0) \\ \mu_1\,:\,\supp(\mu_1) \subseteq g^{-1}(1)}} \cost(\mu_0,\mu_1).
\]
If $g$ is constant, define $\sfR(g)$ to be $0$.
\end{definition}

The result we need from our companion manuscript
\cite{BB20a} can then be phrased as follows.

\begin{theorem}\label{thm:minimax}
For all (possibly partial) Boolean functions $g$,
$\sfR(g)=\Omega(\R(g))$.
\end{theorem}

\section{Counterexample to perfect composition}

To define the partial functions used to prove
\thm{counterexample}, we will use $f(x)=*$ to denote
that $x\notin\Dom(f)$.

\begin{definition}
Define $\gapmaj_m \colon \{0,1\}^m \to \{0,1\}$
to be the gap majority function
\[
\gapmaj_m(x) = \begin{cases}
1 & \mbox{if } |x| = \lceil\frac m2 + 2\sqrt{m}\rceil \\
0 & \mbox{if } |x| = \lfloor\frac m2 - 2\sqrt{m}\rfloor \\
* & \mbox{otherwise.}
\end{cases}
\]
\end{definition}

Note that this is simply the majority function with
a Hamming weight promise which restricts the input
to two Hamming levels $O(\sqrt{m})$ apart.

\begin{lemma}\label{lem:gapmaj}
The randomized query complexity of the gap majority
function on $m$ bits is
\[R(\gapmaj_m) = \Theta(m).\]
\end{lemma}

The proof of this lemma is a standard argument,
but we repeat it here for completeness.

\begin{proof}
The upper bound follows by querying all the bits of the
input. For the lower bound, let $\mu$ be the uniform
distribution on the domain of $\gapmaj_m$.
Suppose there was an algorithm $R$ that solved
$\gapmaj_m$ to error $1/3$ using only $m/1000$ queries.
Then by convexity, there is some deterministic decision
tree $D$ in the support of $R$ that solves $\gapmaj_m$
to bounded error against inputs from $\mu$.
The height of $D$ is still at most $m/1000$.

Now, since $\mu$ is symmetric under permuting
the input bits, the order in which $D$ queries
the inputs doesn't matter; we can assume it reads
them from left to right. Indeed, we can even
assume that $D$ reads the first $k=m/1000$ bits
of the input $x$ in one batch, and then gives
the output. Further, it is not hard to see
that $D$ maximizes its probability of success
by outputting the majority of the $k$ bits it sees.
Assume for simplicity that $k$ is odd.
Then the success probability of $D$ is the same on
$0$- and $1$-inputs from $\mu$, and equals
the probability that, when a string of length
$m$ and Hamming weight $\lceil m/2+\sqrt{m}\rceil$
is selected at random, its first $k$ bits have
Hamming weight at least $k/2$.

The $k$ bits are selected from the
$m$ bit string of that Hamming weight without replacement.
However, if they were selected with replacement,
the probability of seeing at least $k/2$
ones out of the $k$ bits would only increase,
so it suffices to upper bound the probability
of seeing $k/2$ or more ones in a string
of length $k$ when each bit is sampled
independently from $\Bernoulli(1/2+1/\sqrt{m})$.
This is precisely what we get by amplifying
bias $2/\sqrt{m}$ using $m/1000$ repetitions,
which is bias at most $1/5<1/3$ (and hence error
greater than $1/3$) by \lem{amplify}.
This gives a contradiction.
\end{proof}

We will take the inner function $g$ to be $\gapmaj_m$
in our counterexample. This is also essentially the same
inner function as used in the relational
counterexample of \cite{GLSS19}. In that construction,
the outer relation took an $m$ bit string $x$ as input
and accepted as output any string $y$ that has Hamming
distance within $m/2-\sqrt{m}$ of $x$. This relation
requires $\Theta(\sqrt{m})$ queries to solve to bounded
error using a randomized algorithm, but $f\circ g$
can be computed using only $O(m)$ queries instead of
$m^{3/2}$.

Our construction is motivated by this approach, but is somewhat
different as we need $f$ to be a partial function.
Let $\textsc{ApproxIndex} : \{0,1\}^k \times \{0,1,2\}^{2^k} \to \{0,1,*\}$ be the partial function on
$n = k + 2^k$-dimensional inputs defined by
\[
\textsc{ApproxIndex}(a,x) = \begin{cases}
x_a & \mbox{if } x_b = x_a\in\B \mbox{ for all $b$ that satisfy } |b-a| \le \frac k2 - 2\sqrt{k \log k}\\
&\mbox{and } x_b=2 \mbox{ for all other $b$},\\
*   & \mbox{otherwise.}
\end{cases}
\]
In other words, $\textsc{ApproxIndex}$
takes input strings that have two parts: the index
part and the array part. The promise is that
in the array, all positions within $k/2-2\sqrt{k\log k}$
of the index have the same Boolean value, and all
positions far from the index contain the value $2$.
Essentially,
the goal is to find an approximation of the index.

Note that $\textsc{ApproxIndex}$ has input alphabet
of size $3$. We can easily convert this into a function
with input alphabet $\B$ by using binary representation,
which only changes the input size and the complexity
of the function by a constant factor. Hence we will treat
$\textsc{ApproxIndex}$ as a partial Boolean function.
This will be our outer function $f$.
We now show the following lemma.

\begin{lemma}
The randomized query complexity of the approximate address function on $n = k + 2^k$ bits is
\[
R(\textsc{ApproxIndex}) = \Theta(\sqrt{k\log k}) = \Theta(\sqrt{\log n\log\log n}).
\]
\end{lemma}

\begin{proof}
The upper bound is obtained by the simple algorithm that obtains an approximate address $b$ by querying and copying the first $8\sqrt{k \log k}$ bits of $a$ and setting the remaining bits of $b$ uniformly at random, then queries $x_{b}$ and returns that value. The distance $|b-a|$ between the approximate and actual addresses is a random variable with binomial distribution distribution with parameters $N=k-8\sqrt{k\log k}$ and $p=\frac12$ 
so standard tail bounds imply that the algorithm has bounded error.

For the lower bound, we describe a hard distribution.
Let $\mu$ be the distribution over valid inputs to
$\textsc{ApproxIndex}$ which first picks $a\in\B^k$
uniformly at random, then picks a bit $z\in\B$
uniformly at random, and fills the array
with $z$ in positions within $k/2-2\sqrt{k\log k}$
of $a$ and with $2$ in positions further from $a$.
That is, when the distribution picks the pair $(a,z)$,
it generates a valid input whose index part is $a$
and whose function value is $z$.

Suppose there was a randomized algorithm $R$
which solved $\textsc{ApproxIndex}$ to bounded error
using only $\sqrt{k\log k}$ queries.
Then $R$ also solves $\textsc{ApproxIndex}$
against inputs from $\mu$. By convexity,
there is some deterministic decision tree $D$
in the support of $R$ which still computes
$\textsc{ApproxIndex}$ correctly (to bounded error)
against $\mu$, with height at most $\sqrt{k\log k}$.

Consider the deterministic algorithm $D'$
which runs $D$, except whenever $D$ queries inside
the array part of the input, $D'$ does not
make that query and just pretends the answer was $2$.
(Whenever $D$ queries inside the index part of the input,
$D'$ does implement that query correctly.)
Then $D'$ uses at most as many queries as $D$ does,
and never queries inside the array part of the input.
Note that against distribution $\mu$, the success
probability of $D'$ must be exactly $1/2$,
regardless of how its leaves are labeled, because
$\mu$ generates its index (the only part $D'$ queries)
independently from the function value $z$.
So we know $D'$ fails to compute $\textsc{ApproxIndex}$
to bounded error against $\mu$.
Since $D$ succeeds in computing $\textsc{ApproxIndex}$
to bounded error against $D'$, this means that
$D$ and $D'$ output different answers
when run on $\mu$ with constant probability.

Since $D$ and $D'$ behave differently on $\mu$
with constant probability, it means that $D$
has constant probability of querying a non-$2$
position of the array (since in all other cases,
$D'$ behaves the same as $D$). This also means
that if we run $D'$ and look at the set $S$
of array queries it faked the answer to (returning
$2$ instead of making a true query to the array),
then the probability that $S$ contains a non-$2$
position of the array is at least a constant.

To rephrase: we now have an algorithm $D'$
that looks at at most $\sqrt{k\log k}$ positions
of a random string $a$ of length $k$, and returns a
set $S$ of at most $\sqrt{k\log k}$
strings of length $k$ that has a constant
probability of being within $k/2-2\sqrt{k\log k}$
of $a$. By picking a string from $S$ at random,
we can even get an algorithm that looks at
$\sqrt{k\log k}$ positions of $a$ and returns
a string $b$ that has probability at least $1/k$
of being within $k/2-2\sqrt{k\log k}$ of $a$.
This means that of the $k-\sqrt{k\log k}$
positions the algorithm did not look at,
it guessed at least $k/2+\sqrt{k\log k}$
of them correctly with probability at least $1/k$.
But since $a$ is a uniformly random string,
the chance of this happening can be bounded
by the Chernoff bound: it is at most $1/k^2$,
giving the desired contradiction.
\end{proof}

From here,
the proof of \thm{counterexample} is obtained by giving
an upper bound on the randomized query complexity of
the composed function
$\textsc{ApproxIndex} \circ \gapmaj_{\log n}$,
with the $\textsc{ApproxIndex}$ on $n$ bits
(i.e.\ $k=O(\log n)$). If a tight composition theorem
held, the randomized query complexity of this function
would be $\Omega(\log^{3/2}n\sqrt{\log \log n})$. However, there is
an $O(\log n\log\log n)$ randomized query algorithm for this
composed function: the randomized algorithm can first query
$O(\log\log n)$ bits from each of the first $k$ copies of
$\gapmaj_{\log n}$; since this gives it bias
$O(\sqrt{\log\log n}/\sqrt{\log n})$
(i.e.\ $O(\sqrt{\log k}/\sqrt{k})$)
towards the right answer for each bit of $a$
(from \lem{amplify}),
the string of $k$ such bits will
(with high probability) be such that
$|b-a|\le k/2-2\sqrt{k\log k}$. Then the randomized algorithm
can query $x_b$ by using $\log n$ queries to the
appropriate copy of $\gapmaj_{\log n}$, computing
it exactly. This is a total of only
$O(\log n\cdot \log\log n)$ queries instead of
$\Omega(\log^{3/2}n\sqrt{\log\log n})$.

\section{Simulating oracles}

The heart of the proof of \thm{composition} is the \emph{oracle simulation problem} that we describe below.

\paragraph{Oracle simulation problem.} Fix any two (publicly known) distributions $\mu_0$ and $\mu_1$ over $\B^n$.
There is a (true) oracle $\mathcal{O}$ that knows the value of some bit $b \in \B$, samples a string $x \leftarrow \mu_b$, and then provides (noiseless) query access to the bits in $x$. (I.e., on query $i \in [n]$, the oracle returns the value $x_i$.)
In the oracle simulation problem, we do not know $b$, but we wish
to simulate the behavior of $\mathcal{O}$. Our only resource
is a noisy oracle for $b$ as in \defn{noisy_oracle}.
Given access to such a noisy oracle for $b$, our goal is to
simulate $\mathcal{O}$, even in the setting where queries arrive in a stream and we don't know what future queries might be or even when they stop, while minimizing our query cost to the noisy oracle.

Note that we can always solve the oracle simulation problem
by querying $b$ with certainty; that is, we can feed in
$\gamma=1$ into the noisy oracle for $b$, extracting the correct
value of $b$ with probability $1$. Afterwards, we can clearly
use the value of $b$ to match the behavior of $\mathcal{O}$
by generating a sample $x\leftarrow\mu_b$ and using it to answer
queries. The cost of this trivial protocol is $1$ (since
we pay $\gamma^2$ when we go to the noisy oracle with parameter
$\gamma$). Our goal will be to improve this to a cost
that depends on the types of queries made and on the distributions
$\mu_0$ and $\mu_1$, but that in general can be much less than $1$.

\subsection{Simulating a single oracle query}

We first show in this section that the oracle simulation problem can be solved efficiently in the special case where we only have to simulate the true oracle $\mathcal{O}$ for a single query.

\begin{lemma}\label{lem:single_query_simulation}
For any pair $(\mu_0,\mu_1)$ of distributions over $\B^n$,
there is a protocol for the oracle simulation problem such that for 
any single query $i \in [n]$, 
the expected cost of the protocol simulating $i$ is at most
\[
2\Ess^2(\mu_0|_i,\mu_1|_i)
\]
(here $\mu_b|_i$ denotes the marginal distribution of $\mu_b$
onto the bit at index $i$),
and the output of the protocol has \emph{exactly} the same distribution as the output returned by the true oracle on the same query.
\end{lemma}

\begin{algorithm}[ht]
\caption{\textsc{SingleBitSim}($\mu_0,\mu_1,i$)}
\label{alg:single-bit}
  $a \gets \argmin_{c \in \{0,1\}}\{\mu_0|_i(c) + \mu_1|_i(c)\}$\;
  $p_0 \gets \mu_0|_i(a)$\;
  $p_1 \gets \mu_1|_i(a)$\;
  \If{$\Bernoulli(p_0+p_1) = 1$}{
    \If{$\textsc{NoisyOracle}_b\left(\frac{p_0-p_1}{p_0+p_1}\right)$}{
      \Return $a$\;
    }
  }
  \Return $1-a$\;
\end{algorithm}

\begin{proof}
The \textsc{SingleBitSim} algorithm described in \alg{single-bit} returns $a$ if and only if the random variable drawn from the $\Bernoulli(p_0+p_1)$ distribution is $1$ and the \textsc{NoisyOracle} call to $b$ also returns $1$, so
\begin{align*}
\Pr[\textsc{SingleBitSim} \mbox{ returns } a ]
&= (p_0 + p_1) \left( \frac12 + (-1)^b \frac{p_0 - p_1}{2 (p_0+p_1)}\right) \\
&= \frac{p_0+p_1}{2} + (-1)^{b} \frac{p_0 - p_1}{2} 
= p_{b},
\end{align*}
which is also exactly the probability that the true oracle $\mathcal{O}$ returns $a$.

The cost of the algorithm is $0$ with probability $1-(p_0+p_1)$ and $( \frac{p_0-p_1}{p_0+p_1})^2$ otherwise so the expected cost is
\[
(p_0+p_1) \left( \frac{p_0-p_1}{p_0+p_1}\right)^2 
= \frac{(p_0-p_1)^2}{p_0+p_1}
= \frac{(\mu_0|_i(1)-\mu_1|_i(1))^2}{\mu_0|_i(1)+\mu_1|_i(1)}
\le 2\Ess^2(\mu_0|_i,\mu_1|_i). \qedhere 
\]
\end{proof}

\subsection{Simulating multiple queries to the oracle}

We build on the \textsc{SingleBitSim} algorithm to obtain a protocol that simulates any sequence of queries to the true oracle. Let us use $z \in \{0,1,*\}^n$ to denote a partial assignment to a variable $x \in \B^n$; with each coordinate $j \in [n]$ for which $z_j = *$ corresponding to the bits that have not yet been assigned. And for a partial assignment $z$ and a distribution $\mu$ on $\B^n$, we write $\mu^z$ to denote the conditional distribution of $\mu$ conditioned on $z$ being a partial assignment to the sample $x$ drawn from the distribution.

The general \textsc{OracleSim} protocol processes each received query using \textsc{SingleBitSim}, as described in \alg{oracle-sim}. The strategy is to keep calling
\textsc{SingleBitSim} to answer all queries until we see
that the expected total cost of the queries we received exceeds $1$;
at that point, we switch strategies to the trivial
protocol, extracting $b$ with certainty and using it to
answer all further queries.

\begin{algorithm}[ht]
\caption{\textsc{OracleSim}($\mu_0,\mu_1$)}
\label{alg:oracle-sim}
  $z \gets *^n$\;
  $c \gets 0$\;
  \For{each query $i \in [n]$ received}{
    $z_i \gets \textsc{SingleBitSim}(\mu_0^z, \mu_1^z, i)$\;
    Answer the query with $z_i$\;
    \smallskip
    $c \gets c + \h^2( \mu_0^z|_i, \mu_1^z|_i)$\;
    \If{$c > 1$}{
      {\bf break}\;
    }
  }

  \bigskip
  \tcc{If the expected cost of noisy queries exceeds $1$, query the value of $b$ directly to complete the simulation.}
  $b \gets \textsc{NoisyOracle}_b(1)$\;
  \For{each query $i \in [n]$ received}{
    $z_i \gets \mu_b^z|_i$\;
    Answer the query with $z_i$\;
  }
\end{algorithm}

\begin{lemma}
\label{lem:oraclesim-correct}
For any pair $(\mu_0,\mu_1)$ of distributions over $\B^n$ and any sequence of queries,
the distribution of the answers to the queries returned by the \textsc{OracleSim} protocol is identical to the distribution of answers returned by the true oracle on the same sequence of queries.
\end{lemma}

\begin{proof}
This immediately follows from the fact that \textsc{SingleBitSim} answers individual queries with the same distribution as the true oracle.
\end{proof}

In particular, \lem{oraclesim-correct} implies that the behaviour of randomized algorithms does not change when access to the true oracle is replaced with usage of the \textsc{OracleSim} protocol instead.

We now want to bound expected cost of the \textsc{OracleSim} protocol on randomized decision trees. To do so, we must first introduce a bit more notation and establish some preliminary results. For any transcript $\tau = \tran(D,x)$ of a deterministic decision tree $D$ on some input $x$ and any index $t \le |\tau|$, we let $\tau_{< t}$ denote the part $(i_1,x_{i_1}),\ldots,(i_{t-1},x_{t-1})$ of the transcript representing the first $t-1$ queries. That is,
$\tau_{<t}$ is a partial assignment of size $t-1$.

\begin{definition}[Distinguishing distributions]
For any bias $\eta \in (0,1)$, we say that a transcript $\tau$ \emph{$\eta$-distinguishes} two distributions $\mu_0$ and $\mu_1$ if there is an index $t \le |\tau|$ for which a random variable $X \sim \Bernoulli(\frac12)$ satisfies
\[
\left| \E[ X^{\tau_{< t}} ] - \tfrac12 \right| \ge \tfrac\eta2
\]
where $X^{\tau_{< t}}$ is the random variable $X$ conditioned on $\tran(D,\mu_X)_{< t} = \tau_{< t}$.
\end{definition}

In other words, we say a transcript $\tau$
distinguishes two distributions
if at \emph{any point} during the run of $\tau$,
the partial assignment seen up to that point is much more likely
under one of $\mu_0$ or $\mu_1$ than under the other.
We use the following bound on the probability of seeing
a distinguishing transcript $\tau$ when running
an algorithm on the mixture of $\mu_0$ and $\mu_1$.

\begin{lemma}
\label{lem:distinguish}
There exists a constant $\eta \in (0,1)$ such that for every deterministic decision tree $D$ and every pair of distributions $\mu_0$, $\mu_1$ on inputs, when $X \sim \Bernoulli(\frac12)$ then
\[
\Pr_{\tau \sim \tran(D,\mu_X)}[ \tau \mbox{ $\eta$-distinguishes } \mu_0,\mu_1 ] = O\left( \h^2( \tran(D,\mu_0), \tran(D,\mu_1)) \right).
\]
\end{lemma}

\begin{proof}
Let $\rho$ denote the probability that a transcript $\tau$ drawn from $\tran(D,\mu_X)$ $\eta$-distinguishes the distributions $\mu_0$ and $\mu_1$.
We show that $O(1/\rho)$ transcripts sampled independently from the distribution $\tran(D,\mu_b)$ suffice to determine the value $b$ with bounded error. 
The lemma then follows from \clm{h_interpretation}.

The algorithm for determining $b$ given these transcripts
will be Bayesian: it will start with an even prior
on $b=0$ and $b=1$, and then process each sample in turn
-- and within each sample, each query of the transcript in turn --
and update its belief using Bayes' rule. At each point in time,
we keep track of the log odds ratio of the current posterior
distribution. That is, if the belief of the algorithm
is probability $p$ that $b=1$ and probability $1-p$ that
$b=0$, the log odds ratio is defined as $\log(p/(1-p))$.
If at any point in the algorithm, the absolute value of
the log odds ratio exceeds $(1/2)\log((1+\eta)/(1-\eta))$,
the algorithm terminates and returns $1$ if its log odds ratio
is positive and $0$ if its log odds ratio is negative.
If the algorithm reaches the end of all samples without
terminating in this way, it outputs arbitrarily.
In other words, the algorithm reads all the queries of all the
transcripts sequentially, and if ever it reaches very high confidence
of the value of $b$, it outputs that value (and terminates),
but otherwise it guesses randomly when it reaches the end.

To analyze this algorithm, we observe that the log odds ratio
updates additively: if the prior probability that $b=1$ was $p$,
and an event $A$ was observed, the posterior probability
that $b=1$ is $\Pr[b=1|A]=\Pr[A|b=1]\cdot p/\Pr[A]$ and the
posterior probability that $b=0$ is $Pr[A|b=0]\cdot (1-p)/\Pr[A]$,
so their ratio is $p/(1-p)$ times $\Pr[A|b=1]/\Pr[A|b=0]$. It follows
that the posterior log odds ratio is equal to the prior log odds
ratio plus $\log(\Pr[A|b=1]/\Pr[A|b=0])$.

Now, if $\tau$ $\eta$-distinguishes $\mu_0$ and $\mu_1$ 
and if $t$ is such that $X^{\tau_{<t}}$ has bias at least $\eta$,
it means that for this $\tau$, if we were to see $t-1$ queries
starting from an even prior ($0$ log odds ratio), we would
arrive at bias at least $\eta$, meaning the
absolute value of the log odds ratio would
be at least $\log((1+\eta)/(1-\eta))$. Note that this is enough to exceed
the ratio and terminate the algorithm, unless the initial log odds
ratio (before starting reading this transcript $\tau$) was not $0$.
But the only way for the total log odds ratio not to exceed
$\frac12 \log((1+\eta)/(1-\eta))$ in absolute value would be for it
to start at at least $\frac12 \log((1+\eta)/(1-\eta))$ in absolute
value---in which case the algorithm would have terminated before reading
$\tau$! We conclude that reading a transcript that $\eta$-distinguishes $\mu_0$ and $\mu_1$ always causes a termination of this algorithm.

Since we sample $O(1/\rho)$ transcripts, the probability
that we do not see any $\tau$ that $\eta$-distinguishes $\mu_0$ and $\mu_1$ is
$(1-\rho)^{1/\rho}=e^{-\Omega(1)}$,
which we can make an arbitrarily small constant by picking the right
constant in the big-$O$. This means the algorithm always terminates
before reaching the end except with small probability (say, $0.01$),
so it rarely needs to guess.

It remains to argue that when the algorithm terminates, it is
usually correct in its output. Let's suppose $b=0$
(the $b=1$ case is analogous). When the algorithm terminates
and gives an incorrect output,
consider everything it saw up to that point -- this is
some sequence of transcripts plus some sequence of queries
that are part of the transcript causing the termination.
If this sequence is $s$, then the log odds ratio after
observing $s$ must be at least $\frac12 \log((1+\eta)/(1-\eta))$,
meaning the odds ratio must be at least $\sqrt{(1+\eta)/(1-\eta)}$.
In other words, if the probability of seeing such an $s$
when $b=0$ is $p_s$, then the probability of seeing this same $s$
when $b=1$ is at least $p_s\sqrt{(1+\eta)/(1-\eta)}$.
The probability that the algorithm terminates and gives
an incorrect output when $b=0$ is the sum of all such $p_s$;
but then the probability of observing one of those $s$ when
$b=1$ is that sum times $\sqrt{(1+\eta)/(1-\eta)}$. Since
this must be at most $1$, we conclude that
the probability the algorithm terminates and errs when
$b=0$ is at most $\sqrt{(1-\eta)/(1+\eta)}$.
By picking $\eta$ correctly, we can get the probability
of error to be at most $1/3$ (and the $b=1$ case is similar).

By \clm{h_interpretation}, we conclude that
$\rho=O(h^2(\tran(D,\mu_0),\tran(D,\mu_1))$, as desired.
\end{proof}

We are now ready to bound the expected cost of the oracle simulation protocol.

\begin{lemma}
\label{lem:oraclesim-cost}
For any pair $(\mu_0,\mu_1)$ of distributions over $\B^n$, and any randomized decision tree $R$, 
the expected cost of the \textsc{OracleSim} protocol is at most
\[
O\big( \h^2(\tran(R,\mu_0),\tran(R,\mu_1)) \big).
\]
\end{lemma}

\begin{proof}
Note first that $\tran(R,\mu)$ is a disjoint mixture of distributions of the form $\tran(D,\mu)$ for deterministic decision trees $D$, since our definition of the transcript of randomized decision trees includes a copy of the sampled tree $D$ itself. By \clm{h_mix}, it therefore suffices to show that the expected cost of \textsc{OracleSim} on any deterministic decision tree $D$ is
\[
O\big( \h^2(\tran(D,\mu_0),\tran(D,\mu_1)) \big).
\]

We can represent the expected cost of \textsc{OracleSim} on $D$ as
\[\bE_{\tau\sim\tran(D,\mu_b)}[\cost_b(\tau)],\]
where $b$ is the true value of the unknown oracle bit
and where $\cost_b(\tau)$ is defined as the expected cost
of \textsc{OracleSim} conditioned on $\tau$
being the resulting transcript at the end.
This is the correct expression for the
expected cost because we know \textsc{OracleSim}
will generate transcripts $\tau$ from the same distribution
$\tran(D,\mu_b)$ that the true oracle generates them from.

We will use only two properties of $\cost_b(\tau)$.
The first property is that for all $\tau$ and $b\in\B$,
\begin{equation}\label{first_property}
  \cost_b(\tau)
    \le 5\sum_{t=1}^{|\tau|}\h^2\big(\mu_0^{\tau_{<t}}|_{\tau_t},
                        \mu_1^{\tau_{<t}}|_{\tau_t}\big).
\end{equation}
Here we use $\mu_0^{\tau_{<t}}|_{\tau_t}$ to denote the
conditional distribution of $\mu_0$ conditioned on the
partial assignment $\tau_{<t}$, marginalized
to the position queried in the $t$-th entry of $\tau$.
To see that this property holds,
recall that \lem{single_query_simulation}
(combined with \clm{h_vs_JS_vs_S})
provides an upper bound of
$4\h^2\big(\mu_0^{\tau_{<t}}|_{\tau_t},
                        \mu_1^{\tau_{<t}}|_{\tau_t}\big)$
on the cost of query $t$ of $D$
conditioned on $\tau_{<t}$ being seen previously,
unless query $t$ causes a cutoff which forces a cost of $1$.
By the definition of \textsc{OracleSim},
this cutoff only happens if the sum
$\sum_{i=1}^{t-1}\h^2\big(\mu_0^{\tau_{<i}}|_{\tau_i},
                        \mu_1^{\tau_{<i}}|_{\tau_i}\big)$
(which is stored in variable $c$) exceeds $1$;
in this case, the cutoff only causes the sum over $t$
of $4\h^2\big(\mu_0^{\tau_{<t}}|_{\tau_t},
                        \mu_1^{\tau_{<t}}|_{\tau_t}\big)$
to increase by at most a factor of $5/4$,
since before the cutoff it must already have been at least $4$.

The second property we will need is that for all $\tau$ and $b$,
\begin{equation}\label{second_property}
  \cost_b(\tau)\le 10.
\end{equation}
This follows from the first property by noticing
that if a cutoff is reached, no further queries are made,
and the variable $c$ can at most exceed its cutoff $1$
by $1$ (since $\h^2$ is always bounded above by $1$).

Our goal is to upper bound the expected cost of \textsc{OracleSim},
which we know can be written
$\bE_{\tau\sim\tran(D,\mu_b)}[\cost_b(\tau)]$,
by $O(\h^2(\tran(D,\mu_0),\tran(D,\mu_1)))$.
We start by noting that the latter expression can be lower bounded using \clm{h_vs_JS_vs_S} and \clm{JS_I}:
\[
\h^2(\tran(D,\mu_0),\tran(D,\mu_1))
\ge \tfrac12 \JS(\tran(D,\mu_0),\tran(D,\mu_1))
= \tfrac12 I(X;\tran(D,\mu_X)).
\]
Using the chain rule for mutual information and the definition of conditional information, we then obtain
\begin{align*}
\h^2(\tran(D,\mu_0),\tran(D,\mu_1))
&\ge \tfrac12 \sum_{t=1}^n I(X;\tran(D,\mu_X)_t \mid \tran(D,\mu_X)_{<t})\\
&=\tfrac12 \sum_{t=1}^n\bE_{\tau\sim\tran(D,\mu_X)}
[I(X^{\tau_{<t}};\tran(D,\mu_X)^{\tau_{<t}}_t)],
\end{align*}
where as usual we use $\tau_{<t}$ to denote the transcript $\tau$
cut off before query $t$ (meaning that the sequence
$(i_1,x_{i_1}),(i_2,x_{i_2}),\dots$ in the transcript
gets truncated after $(i_{t-1},x_{i_{t-1}})$),
and $\tau_t$ to denote query $t$ of the transcript
(meaning the single pair $(i_t,x_{i_t})$ in position
$t$ of the sequence).
We can exchange the sum and the expectation
and we can also replace $n$ by
$|\tau|$ as the information of the transcript is always $0$ after
the transcript ends. Doing so yields
\[
\h^2(\tran(D,\mu_0),\tran(D,\mu_1))
\ge 
\tfrac12 \bE_{\tau\sim\tran(D,\mu)}
\left[\sum_{t=1}^{|\tau|}
I(X^{\tau_{<t}};\tran(D,\mu_X)^{\tau_{<t}}_t)\right],
\]
where we are using $\mu$ to denote $\mu_X=(1/2)(\mu_0+\mu_1)$.

Let $S_1$ denote the set of transcripts $\tau$ that do not $\eta$-distinguish $\mu_0$ and $\mu_1$, and $S_2$ be the other transcripts (that do $\eta$-distinguish $\mu_0$ and $\mu_1$) for the value of $\eta$ guaranteed to exist by \lem{distinguish}.
We write $\tau\sim S_1$ to mean $\tau$ sampled from the conditional
distribution $\tran(D,\mu)$ conditioned on $\tau\in S_1$,
and similarly for $\tau\sim S_2$. Then
\begin{align*}
\h^2(\tran(D,\mu_0)\tran(D,\mu_1))
&\ge \tfrac12 \Pr[\tau\in S_1]\cdot
\bE_{\tau\sim S_1}\left[\sum_{t=1}^{|\tau|}
I(X^{\tau_{<t}};\tran(D,\mu_X)^{\tau_{<t}}_t)\right]\\
&\ge \frac{1-\eta}{2}\Pr[\tau\in S_1]\cdot\bE_{\tau\sim S_1}
\left[\sum_{t=1}^{|\tau|}
\h^2(\tran(D,\mu_0)^{\tau_{<t}}_t,\tran(D,\mu_1)^{\tau_{<t}}_t)\right],
\end{align*}
where the first line follows by removing the part of
the expectation over $S_2$ (which is non-negative), and the second
line follows from \clm{JS_bias}
(converting a biased coin into an unbiased coin with
$(1-\eta)$ loss) together with \clm{h_vs_JS_vs_S} (converting $\JS$
distance to $\h^2$). 

Now, observe that each term of the sum
is exactly
$\h^2(\mu_0^{\tau_{<t}}|_{\tau_t},\mu_1^{\tau_{<t}}|_{\tau_t})$.
Hence by \ref{first_property}, we have
\[
\h^2(\tran(D,\mu_0)\tran(D,\mu_1))
 = \Omega \left(\Pr[\tau\in S_1]\cdot\bE_{\tau\sim S_1}
 [\cost_b(\tau)]\right).
\]
We now write
\begin{align*}
\h^2(\tran(D,\mu_0)\tran(D,\mu_1)) 
&=\Omega\left(\bE_{\tau\sim\tran(D,\mu)}[\cost_b(\tau)]
-\Pr[\tau\in S_2]\cdot\bE_{\tau\sim S_2}[\cost_b(\tau)]\right)\\
&=\Omega\left(\bE_{\tau\sim\tran(D,\mu)}[\cost_b(\tau)]\right)
-O(\Pr[\tau\in S_2]),
\end{align*}
where we used \ref{second_property} in the last line.
Finally, since $\mu=(\mu_0+\mu_1)/2$, the expectation
of a nonegative random variable against $\mu_b$ is at most
twice the expectation of that variable against $\mu$.
We have thus obtained that the expected cost \textsc{OracleSim}
is bounded above by
\[O\left(\h^2(\tran(D,\mu_0),\tran(D,\mu_1))
    +\Pr[\tau\in S_2]\right),
\]
and the desired bound $O\big(\h^2(\tran(D,\mu_0),\tran(D,\mu_1))\big)$
follows from \lem{distinguish}.
\end{proof}

\section{The composition theorem}

\subsection{The proof}

Equipped with \lem{oraclesim-correct} and \lem{oraclesim-cost},
we are ready for the proof of \thm{composition}.
In fact, we prove a slightly stronger version of the theorem:
we show that the hard distribution for $f\circ g$
can be assumed to take the form of a distribution of $f$
composed with a distribution of $g$.
Start with the following definitions.

\begin{definition}
Let $\mu_0$ and $\mu_1$ be distributions over
$\B^m$,
and let $y\in\B^n$. Then
define $\mu_y\coloneqq\bigotimes_{i=1}^n\mu_{y_i}$,
which is a distribution over $\Sigma^{nm}$.
If $\nu$ is a distribution over $\B^n$, define $\nu\circ(\mu_0,\mu_1)$
to be the distribution which samples $y\leftarrow\nu$
and then returns a sample from $\mu_y$.
\end{definition}

\begin{definition}
Let $g$ be a (possibly partial) Boolean function from
a subset of
$\B^m$ to $\B$,
and let $f$ be a function or relation from a subset of $\B^n$
to $\Sigma_O$ (a finite alphabet). Then define
$\compR(f,g)$ to be the maximum, over distributions
$\mu_0$ and $\mu_1$ on $0$-inputs and $1$-inputs of $g$,
of the complexity
of solving $f$ on distributions of the form $\mu_y$ for $y\in\Dom(f)$.
In other words,
\[\compR(f,g)=\max_{\mu_0,\mu_1}\min_R\max_y\cost(R,\mu_y),\]
where $R$ is a randomized algorithm that is required
to compute $f(y)$ with bounded error against all input
distributions of the form $\mu_y$,
and where $\cost(R,\mu_y)$
is the expected number of queries $R$ makes against
distribution $\mu_y$.
We will further write $\compR_\epsilon(f,g)$
when we need to specify the error parameter.
\end{definition}

We note that $\compR(f,g)$ satisfies a minimax theorem
with respect to the minimization over $R$ and the maximization
over $y$. Hence, we can define it as the maximum
randomized query complexity
of a hard distribution for $f\circ g$
which has the form
$\nu\circ(\mu_0,\mu_1)$, with $\nu$ a distribution over $\Dom(f)$
and $\mu_b$ being distributions over $g^{-1}(b)$ for $b\in\B$.
It is also clear that
$\R_\epsilon(f\circ g)\ge \bar{\R}_\epsilon(f\circ g)
\ge\compR_\epsilon(f,g)$,
where $\bar{\R}(f)$ denotes the expected randomized query
complexity of $f$ (against worst-case inputs).

The following theorem implies \thm{composition}
when combined with \thm{minimax}.

\begin{theorem}\label{thm:composition2}
Let $f$ be a partial function or relation on $n$ bits, with Boolean input alphabet and finite output alphabet $\Sigma_O$. 
Let $g$ be a partial Boolean function on $m$ bits. 
Let $\epsilon\in(0,1/2)$. Then
\[
\bar{\R}_\epsilon(f\circ g)\ge\compR_\epsilon(f,g)
= \Omega\left(\noisyR_{\epsilon}(f) \cdot \sfR(g) \right).
\]
\end{theorem}

\begin{proof}
Only the second part needs proof
($\compR(f,g)$ is by definition at most
$\bar{\R}(f\circ g)$).
The idea of the proof is to convert an algorithm for $f\circ g$ into
an algorithm for $f$ that acts on a noisy oracle, thereby
upper bounding $\noisyR(f)$ in terms of $\compR(f\circ g)$.
To do so, we will use the \textsc{OracleSim} protocol $n$ times
to simulate an oracle for each $g$-input.
Recall that \textsc{OracleSim}
allows us to pretend to have a sample $x$ from distribution
$\mu_b$ without knowing $b$ (so long as we have access to a noisy
oracle for $b$). We will use this protocol to run the algorithm
for $f\circ g$ without actually having the $n$ input strings to
the copies of $g$; instead, we will only have noisy oracles for
the $n$ bits to which the copies of $g$ evaluate. This will
define a $\noisyR(f)$ algorithm.

Let $\mu_0$ and $\mu_1$ be hard distributions for $\sfR(g)$,
so that their support is over $g^{-1}(0)$ and $g^{-1}(1)$, respectively, and every randomized decision tree $R$ satisfies
\begin{equation}
\label{eq:cost-lb}
\min\{\cost(R,\mu_0),\cost(R,\mu_1)\}\ge
\sfR(g)\cdot\h^2(\tran(R,\mu_0),\tran(R,\mu_1)).
\end{equation}

Next, consider a randomized algorithm $A$ which solves $f\circ g$
to error $\epsilon$ against distributions $\mu_y$ for $y\in\Dom(f)$
using at most $\compR_\epsilon(f\circ g)$ expected queries.
We will use algorithm $A$ to define an algorithm $B$ which solves $f$
when it accesses the input to $f$ with a noisy oracle.
The algorithm $B$ works as follows. Given noisy-oracle query access
to an input string $y$ of length $n$, the algorithm $B$ creates $n$
instances of the \textsc{OracleSim} protocol.
It instantiates each of those protocols using the distributions $\mu_0$, and $\mu_1$. 
Call these protocol instances $\Pi_1,\Pi_2,\dots,\Pi_n$. 
The algorithm $B$ also hooks up each $\Pi_i$ with the noisy oracle
for $y_i$.
Finally, with these protocols all set up, the algorithm $B$ will simulate the algorithm $A$, and whenever $A$ makes an input to bit number $j$ inside the $i$th copy of $g$ the algorithm
$B$ will feed in query $j$ into $\Pi_i$ and then return to $A$
whatever alphabet symbol $\Pi_i$ returns. 
When $A$ terminates, the algorithm $B$ outputs the output of $A$.

We analyze the correctness of $B$ on an arbitrary input $y\in\Dom(f)$.
We know that $A$ correctly solves $f\circ g$ on $\mu_y$
to error $\epsilon$.
By \lem{oraclesim-correct}, the protocols $\Pi_i$
act the same as the true oracles. 
So the error of $B$ is also at most $\epsilon$.

Next, we wish to show that the expected cost of the queries $B$
makes on an arbitrary input $y$ is at most
$O(\compR_\epsilon(f,g)/\sfR(g))$.
To start, we note that the simulation of $A$ that $B$ runs
makes at most $\compR_\epsilon(f,g)$ queries in expectation.
Now, for each $i$, let $A_y^i$ be the algorithm $A$
restricted to make queries only in the $i$-th input to $g$,
with all other inputs generated artificially from their fake oracles;
that is, $A_y^i$ is a algorithm acting on only $m$ bits,
which sets up $n-1$ fake oracles and runs $A$ on the fake oracles
with the true input in place of the $i$-th oracle.
Then the expected number of queries $A$ makes against $\mu_y$
is $\sum_{i=1}^n\cost(A_y^i,\mu_{y_i})$, so this sum is
at most $\compR_\epsilon(f,g)$.

We wish to bound the cost of $B$, which is the expected number
of queries all the protocols $\Pi_i$ make to the noisy oracles.
By \lem{oraclesim-cost},
the expected cost of the noisy queries to the noisy oracle for $y_i$
made by the protocol $\Pi_i$ when implementing $A$
is at most $C\h^2(\tran(A_y^i,\mu_0),\tran(A_y^i,\mu_1))$
for some constant $C$, so
the total cost of $B$ on $\mu_y$ is at most
$C\sum_{i=1}^n \h^2(\tran(A_y^i,\mu_0),\tran(A_y^i,\mu_1))$.
Furthermore, by~\eqref{eq:cost-lb}, for every $i$ we have
\[
\cost(A_y^i,\mu_{y_i}) \ge \sfR(g) \h^2(\tran(A_y^i,\mu_0),\tran(A_y^i,\mu_1))
\]
and so the expected cost of $B$ is bounded above by
\[
C\sum_{i=1}^n \h^2(\tran(A_y^i,\mu_0),\tran(A_y^i,\mu_1))
\le C\frac1{\sfR(g)} \sum_{i=1}^n \cost(A_y^i,\mu_{y_i}) 
\le C\frac{\compR_\epsilon(f,g)}{\sfR(g)}.
\]
This shows that
\[
\noisyR_\epsilon(f) \le
C \frac{\compR_\epsilon(f,g)}{\sfR(g)},
\]
as desired.
\end{proof}

\subsection{Further discussion}

Our phrasing of the composition theorem
in terms of $\compR(f,g)$ highlights the fact
that our composition theorem is \emph{distributional}:
it constructs a hard distribution for $f\circ g$
using a hard distribution for $f$ and a hard
distribution for $g$. This is not unique
to our work; most composition theorems in the literature
seem to be distributional in this way,
though this is not usually emphasized.

One interesting thing about distributional
composition theorems is that they are not
obvious even when the outer function is trivial.
For example, consider the function
$\triv_n$, a promise problem on $n$ bits
whose domain is $\{0^n,1^n\}$ and which
maps $0^n\to 0$ and $1^n\to 1$.
We have $\R(\triv_n)=1$.
It is also immediately clear that
$\R(\triv_n\circ g)=\Omega(\R(g))$,
because if we give each copy of $g$ the
same input $x$, computing $\triv_n\circ g$
is equivalent to computing $g$ on $x$.
However, this lower bound on $\R(\triv_n\circ g)$
is not distributional! That is to say,
the hard distribution implicit in this argument
for $\triv_n\circ g$ does not have the form
of a hard distribution for $\triv_n$ composed
with a hard distribution for $g$.

Indeed, the question of proving a
\emph{distributional} composition
theorem for $\triv_n\circ g$
(that is, the problem of lower bounding
$\compR(\triv_n,g)$)
is what is called the correlated
copies problem in
the concurrent work of \cite{BDG+20}.
They prove $\compR(\triv_n,g)=\Omega(\R(g))$.
This is also matched by our
independent composition theorem above,
since we show
$\compR(\triv_n,g)=\Omega(\noisyR(\triv_n)\R(g))$
and since $\noisyR(\triv_n)=\Omega(1)$
(see \lem{noisyR_one}).

\section{Characterizing \texorpdfstring{$\noisyR(f)$}{noisyR(f)}}

In this section we characterize $\noisyR(f)$
as $\R(f\circ\gapmaj_n)/n$. We also show
that in the non-adaptive setting, $\noisyR(f)$
and $\R(f)$ are equal up to constant factors.

\subsection{Warm-up lemmas}

To start, we show that a $\noisyR(f)$ algorithm
can always be assumed to use only two bias parameter
settings: either bias $1$ or an extremely small bias.

\begin{lemma}\label{lem:two_biases}
For any (possibly partial) Boolean function
$f$, there is a randomized
algorithm for $f$ on noisy oracles which has
worst-case expected cost $O(\noisyR(f))$,
but which only queries its noisy
oracles with parameter either $\gamma=1$
or $\gamma=\hat{\gamma}$ (for a single
value of $\hat{\gamma}>0$ that may depend on $f$).
\end{lemma}

\begin{proof}
Let $A$ be a noisy oracle algorithm for $f$
with cost at most $2\noisyR(f)$.
Recall that noisy oracle algorithms are finite
probability distributions over finite decision
trees, so there are finitely possible
queries to a noisy oracle that $A$ can ever make.
Out of those finitely many possible queries,
let $\hat{\gamma}$ be the smallest nonzero
bias parameter that $A$ ever uses.
We now construct a noisy oracle algorithm $B$
that only makes queries with bias parameter
$\gamma=\hat{\gamma}$ or $\gamma=1$.

The algorithm $B$ works by simulating $A$.
If $A$ makes a query to a noisy oracle
with parameter $\gamma\in[1/3,1]$,
the algorithm $B$ simulates this query
by using parameter $\gamma=1$ instead,
and then artificially adding exactly the right
amount of noise to match the behavior of $A$.
The cost $B$ incurs in making such a query is $1$,
but the cost that $A$ incurred was at least
$1/9$, so this is only a factor of $9$ larger.
This covers all queries $A$ makes with parameter
$\gamma\ge 1/3$.

If $A$ makes a query with parameter
$\gamma\in[\hat{\gamma},1/3)$,
the algorithm $B$ will make
$O(\gamma^2/\hat{\gamma}^2)$ queries
with parameter $\hat{\gamma}$ and take their
majority vote. By \lem{amplify},
this will provide $B$ with a bit $\tilde{b}$ that
has bias greater than $\gamma$ towards
the true value of the input bit. The
algorithm $B$ will then add additional noise
to $\tilde{b}$ in order to decrease its
bias to precisely $\gamma$, matching
the behavior of $A$.
The cost incurred by $B$ in this simulation
is $O(\gamma^2/\hat{\gamma}^2)\cdot \hat{\gamma}^2$,
which is $O(\gamma^2)$, matching the cost
incurred by $A$ up to a constant factor.
\end{proof}

We note that the above lemma also works
when $f$ is a relation, and also works
for $\noisyR_\epsilon(f)$ for any
error parameter $\epsilon$.

Our next lemma shows that $\noisyR(f)$ is always
at least $\Omega(1)$ when $f$ is non-trivial;
in particular, if the input has a sensitive block,
a $\noisyR(f)$ algorithm must make
queries of cost $\Omega(1)$ within that block.

\begin{lemma}\label{lem:noisyR_one}
Let $x,y\in\B^n$ be strings which differ on the
block $B\subseteq[n]$. Then any noisy oracle
algorithm $A$ which distinguishes $x$ from $y$
with bounded error must,
when run on either $x$ or $y$, make queries inside $B$
of total expected cost $\Omega(1)$.

In particular, if $f$ is a
(possibly partial) Boolean function
that is not constant, then $\noisyR(f)=\Omega(1)$.
This also applies to relations $f$ that
have two inputs $x,y$ with disjoint
allowed output sets.
\end{lemma}

\begin{proof}
If $f$ is not constant, there exist some
$x,y\in\Dom(f)$ with $f(x)\ne f(y)$.
Any algorithm which computes $f$ can therefore
be used to distinguish $x$ from $y$
with bounded error. Let $B\subseteq[n]$
be the set of indices $i$ for which $x_i\ne y_i$.
Then note that any noisy oracle calls
to noisy oracles for bits outside of $B$
do not help in distinguishing $x$ from $y$.
This reduces the second part of the lemma to the
first part.

Suppose we had a noisy oracle algorithm
distinguishing $x$ and $y$ to bounded error.
Now, up to possible negation, a noisy oracle
call to the $i$-th bit is equivalent
to a noisy oracle call to the $j$-th bit,
since either $x_i=x_j$ and $y_i=y_j$,
or else $x_i=1-x_j$ and $y_i=1-y_j$.
This means that all noisy oracle calls may
as well be made to a single bit $i\in B$.

By \lem{two_biases}, we may assume that a
noisy oracle algorithm distinguishing $x$ from $y$
makes only noisy oracle queries with parameter
$\hat{\gamma}$ or $1$. Let $A$ be such an algorithm,
and we assume that $A$ only ever queries
a single bit of the input. If $A$ ever
uses noisy oracle query with parameter $1$,
it has distinguished $x$ from $y$ with certainty,
so we can halt it there without any loss in our
success probability. Next, we can use
\lem{amplify} to replace the noisy oracle
calls with parameter $1$ with $O(1/\hat{\gamma}^2)$
noisy oracle calls of parameter $\hat{\gamma}$;
doing so decreases the success probability
of $A$ by at most a small additive constant,
and changes the cost of $A$ by at most a constant
factor.

We've reduced to the case where $A$ only
makes noisy oracle queries to a single bit
of the input, all with the same parameter
$\hat{\gamma}$. Let $T_0$ be the expected number
of such calls $A$ makes when run on $x$
and let $T_1$ be the expected number of such calls
it makes when run on $y$, so that its
expected cost is $T_0\hat{\gamma}^2$
and $T_1\hat{\gamma}^2$ respectively.
Assume without loss of generality that
$T_1\ge T_0$. We can cut off the algorithm
$A$ if it ever makes more than $10T_0$ noisy
oracle queries, and have $A$ declare
that the input was $y$; this does not decrease
the success probability of $A$ on input $y$.
Also, on input $x$, a cutoff happens with probability
at most $1/10$ (by Markov's inequality), so
this modification
it changes the success probability of $A$
by at most $1/10$ on input $x$. Hence
this modified algorithm still distinguishes
$x$ from $y$ to bounded error.

Finally, we can replace $A$ with a non-adaptive
algorithm $A'$ which makes $10T_0$ queries
to the oracle with bias $\hat{\gamma}$
all in one batch, and then uses those query
answers to simulate a run of $A$ (feeding
them to $A$ as $A$ requests them). At the end,
$A'$ outputs what $A$ outputs. Then since $A$
distinguishes $x$ from $y$ with constant
probability, so does $A'$, which means
that $A'$ can be used to take $10T_0$
bits of bias $\hat{\gamma}$ and amplify
them to a bit of constant bias.
However, it should be clear that the
best way to take $10T_0$ bits of bias
$\hat{\gamma}$ and output a single bit
with maximal bias is to output the majority
of those bits (this is because if we start
with prior $1/2$ on whether the bits are biased
towards $0$ or $1$, the posterior after
seeing the $10T_0$ bits will lean towards
the majority of the bits).
So the existence of $A'$ ensures we can
take $10T_0$ bits with bias $\hat{\gamma}$,
and their majority will have constant bias.

Finally, by \lem{amplify}, this means that
$10T_0=\Omega(1/\hat{\gamma}^2)$,
which means that the cost of $A$ is $\Omega(1)$,
as desired.
\end{proof}

Finally, we prove the following simple lower bound
on $\noisyR(f)$.

\begin{lemma}\label{lem:noisyR_fbs}
Let $f$ be a (possibly partial) Boolean function.
Then $\noisyR(f)=\Omega(\fbs(f))$.
\end{lemma}

\begin{proof}
Fix input $x\in\Dom(f)$ and sensitive block
$B\subseteq[n]$ for $f$ at $x$. Note that
by \lem{noisyR_one}, any noisy oracle algorithm $A$
computing $f$ must, on input $x$,
make queries inside $B$ of total expected cost
at least $\Omega(1)$. For each bit $i$ of $x$,
let $p_i$ be the total expected cost $A$ makes
to the oracle for $x_i$ when run on $x$. Then
we have $\sum_{i\in B} p_i=\Omega(1)$ for every
sensitive block $B$ for $x$.

Now suppose
$A$ achieves worst-case expected cost $O(\noisyR(f))$,
let $x$ be such that $\fbs_x(f)=\fbs(f)$,
and let $\{w_B\}$ be a feasible weighting scheme
over sensitive blocks $B$ such that
$\sum_B w_B=\fbs_x(f)$.
Then for some constant $C$,
\[C\cdot \noisyR(f)\ge\sum_{i=1}^n p_i\ge
\sum_{i=1}^n p_i\sum_{B:i\in B} w_B
=\sum_B w_B\sum_{i\in B} p_i
=\Omega\left(\sum_B w_B\right)
=\Omega(\fbs(f)).\qedhere\]
\end{proof}

\subsection{Characterization in terms of composition
with gap majority}

We now tackle the task of proving
$\noisyR(f)=\Theta(\R(f\circ\gapmaj_n)/n)$.
The core of the proof will be the following theorem,
which states that $\hat{\gamma}$ in \lem{two_biases}
can be taken to be $1/\sqrt{n}$ without loss of
generality.

\begin{theorem}\label{thm:sqrt_bias}
Let $f$ be a (possibly partial) Boolean function
on $n$ bits. Then there is a noisy oracle algorithm
for $A$ of worst-case expected cost $O(\noisyR(f))$
which uses only noisy oracle queries with parameter
$\gamma=1/\sqrt{n}$ or $\gamma=1$.

This also holds when $f$ is a relation, so long
as there are two inputs $x,y$ that have dijoint
allowed output sets.
\end{theorem}

\begin{proof}
By \lem{two_biases}, there is a noisy oracle
algorithm $A$ for $f$ of worst-case expected
cost at most $O(\noisyR(f))$ which
uses only noisy oracle queries with bias
$1$ or $\hat{\gamma}$.
We will simulate $A$ with a noisy oracle algorithm
$B$ which uses only parameters $1$ or $1/\sqrt{n}$.

Clearly, we can simulate the bias $1$
calls of $A$ with bias $1$ calls in $B$,
so we only need to worry about simulating
the parameter $\hat{\gamma}$ calls. If
$\hat{\gamma}\ge 1/\sqrt{n}$, we can use
multiple noisy oracle calls with parameter
$1/\sqrt{n}$ to simulate one
call with parameter $\hat{\gamma}$
using \lem{amplify}, just like we did
in the proof of \lem{two_biases}.
So the only remaining case is where
$\hat{\gamma}<1/\sqrt{n}$. We can also
assume $f$ is not constant, as the theorem
is easy when $f$ is constant.
For convenience, we will write $\gamma$
in place of $\hat{\gamma}$ from now on,
and we will let $\delta=1/\sqrt{n}>\gamma$.

The idea is to use a single call of bias $\delta$ to generate
a large number of independent bits of bias $\gamma$ each.
The number of bits generated by one call will itself be random,
but we would like its expectation to be
$\Omega(\delta^2/\gamma^2)$.

To achieve this, we note that the sequence of independent
bits that a bias-$\gamma$ oracle should return can
be viewed as a random walk on a line, where each $1$
bit walks forward and each $0$ bit walks backwards.
Let $t=\lfloor \delta/5\gamma\rfloor$, and imagine
placing a mark on the line every $t$ steps in both
directions; that is, positions $0,t,-t,2t,-2t,3t,-3t,\dots$
will all be marked. Note that if the random walk is currently
at one marked point $at$ for some integer $a$,
then with probability $1$, it will eventually reach either
$(a-1)t$ or $(a+1)t$. We generate sequences of steps
in batches: starting from position $at$, we generate
bits until either position $(a+1)t$ or $(a-1)t$ is reached.

To generate such a batch of bits, we first generate
a single bit of from the noisy oracle of bias $\delta$,
and add a small amount of noise to it to decrease its bias
to $\delta'$ (to be chosen later).
If this bit comes out $0$,
we generate a sequence of bits of bias
$\gamma$ conditioned on this sequence reaching $(a-1)t$
before it reaches $(a+1)t$; alternatively, if the bit
is $1$, we generate a sequence of bits of bias
$\gamma$ conditioned on this sequence reaching $(a+1)t$
before it reaches $(a-1)t$.

The first crucial observation is that the distributions
of these sequences are the same whether the bias $\gamma$
is in the $0$ direction or the $1$ direction; that is,
conditioned on reaching $(a+1)t$ before reaching
$(a-1)t$, the probability of each sequence of steps
is identical in the case where the bias is $\gamma$
and in the case where the bias is $-\gamma$.
To see this, pick any such sequence of steps;
say there are $w$ steps forward and $z$ steps back,
with $w-z=t$. The probability of exactly this sequence
occurring is exactly
\[\left(\frac{1+\gamma}{2}\right)^w
\left(\frac{1-\gamma}{2}\right)^z
=\left(\frac{1-\gamma^2}{4}\right)^z
\left(\frac{1+\gamma}{2}\right)^t\]
if the bias is $\gamma$, and exactly
\[\left(\frac{1-\gamma^2}{4}\right)^z
\left(\frac{1-\gamma}{2}\right)^t\]
if the bias is $-\gamma$. Hence the ratio between
the probability under bias $\gamma$ and under bias $-\gamma$
is always
$R\coloneqq\left(\frac{1+\gamma}{1-\gamma}\right)^t$,
which is independent of the sequence of steps.
In other words, for every sequence of steps that ends up
at $(a+1)t$, that sequence is exactly $R$ times more likely
when the bias is $\gamma$ compared to when it is $-\gamma$.
This means that when we \emph{condition} on some
subset of sequences that all reach $(a+1)t$,
the conditional probability will be the same regardless
of whether the bias is $\gamma$ or $-\gamma$.

Now, what is the probability of reaching $(a+1)t$ before
reaching $(a-1)t$? If this probability is $p$
when the bias is $-\gamma$, then it is $R\cdot p$
when the bias is $\gamma$. By symmetry, the probability
of reaching $(a-1)t$ before $(a+1)t$ will be $R\cdot p$
when the bias is $-\gamma$ and $p$ when the bias
is $\gamma$. Since the probability of never reaching either
of $(a-1)t$ or $(a+1)t$ is $0$, we must therefore have
$p+Rp=1$, or $p=1/(R+1)$. That is, the probability
of reaching the threshold in the direction the bias points
towards is $R/(R+1)$, and the probability of reaching
the threshold in the other direction is $1/(R+1)$,
where $R=\left(\frac{1+\gamma}{1-\gamma}\right)^t$.

We pick $\delta'$ so that the probability of a single
bit of bias $\delta'$ being correct is exactly $R/(R+1)$,
and the probability the bit is wrong is $1/(R+1)$.
To do so, we set $(1-\delta')/2=1/(1+R)$, or
$\delta'=(R-1)/(R+1)$. It next will be useful to place some
bounds on $R$.

It is not hard to check using elementary calculus that
$(1+2\gamma/(1-\gamma))^t \ge1+2\gamma t$
holds whenever $t\ge 1$ and $\gamma\in(0,1)$.
We therefore have $R\ge 1+2\gamma t$.
Note that $t=\lfloor \delta/5\gamma\rfloor>\delta/5\gamma-1$
and that $\delta/\gamma>10$; this means
$t>\delta/10\gamma$, so $R\ge 1+\delta/5$.

In the other direction, note that
\[\ln R=t(\ln(1+\gamma)-\ln(1-\gamma))
=2t(\gamma+\gamma^3/3+\gamma^5/5+\dots)
\le 2t\gamma/(1-\gamma^2)).\]
Using $\gamma<1/10$, we have
$\ln R< (5/2)t\gamma$, or $R\le e^{(5/2)t\gamma}$.
Note that for all $x\in[0,1/2]$, we have
\[e^x\le 1+2x.\]
Since $t\le \delta/5\gamma$, we have
$(5/2)t\gamma< \delta/2\le 1/2$, so we have
$R\le e^{(5/2)t\gamma}\le e^{\delta/2}\le 1+\delta$.
Hence $(R-1)/(R+1)=1-2/(R+1)$ is at least
$1-2/(2+\delta/5)=1-1/(1+\delta/10)\ge\delta/5$
and at most $1-2/(2+\delta)=1-1/(1+\delta/2)\le \delta/2$.
Thus our choice of $\delta'$ is smaller than $\delta$
but within a constant factor of $\delta$, so we can easily
convert from a bit of bias $\delta$ to a bit of bias
$\delta'$ by adding noise.

In summary, we can generate a random walk of bias $\gamma$
by first generating the sequence of marked spots
(i.e.\ multiples of $t$) that this sequence visits
as a random walk of bias $\delta'$, and then generating
the sequence of steps that get from a given multiple of $t$
to the subsequent one from the conditional distribution
(which turns out to be the same distribution
regardless of whether the bias is $\gamma$ or $-\gamma$).
This reproduces the correct distribution
over random walks except for probability mass of $0$
(in the cases where the random walk ``gets stuck'' between
$at$ and $(a+1)t$ forever), and probability mass $0$
does not matter to us as our algorithm is finite.

The above is a valid way of simulating noisy oracle
calls to bias $\gamma$ using noisy oracle calls to bias
$\delta>\gamma$. What remains is to analyze the cost
of this procedure.
Note that the expected number
of steps of bias $\gamma$
taken from $at$ until either $(a-1)t$ or
$(a+1)t$ is reached is (by \cite{Fel57}, section XIV.3,
page 317) exactly
\[\frac{t}{\gamma}
\left(1-2(1-\gamma)^t\frac{(1+\gamma)^t-(1-\gamma)^t}
    {(1+\gamma)^{2t}-(1-\gamma)^{2t}}\right).\]
We now lower bound this. Note that
$(1+\gamma)^{2t}-(1-\gamma)^{2t}\ge 4\gamma t$,
and that
\[(1+\gamma)^t-(1-\gamma)^t
=2(\binom{t}{1}\gamma+\binom{t}{3}\gamma^3+\dots)
\le 2(\gamma t+\gamma^3 t^3+\dots)
\le 2\gamma t/(1-\gamma^2 t^2).\]
Also, $(1-\gamma)^t\le 1-\gamma t$. Hence the expectation
is at least
\[\frac{t}{\gamma}
\left(1-\frac{1-\gamma t}{1-\gamma^2 t^2}\right)
=\frac{ t^2}{1+\gamma t}
\ge \frac{t^2}{1+\delta/5}\ge \frac{\delta^2}{120\gamma^2}.\]

In other words, for each call to the oracle of bias $\delta$
(which costs us $\delta^2$), we expect to generate
at least $\delta^2/120\gamma^2$ random bits of bias $\gamma$
(which cost the old algorithm $\gamma^2$ each).
This is exactly what we need, except for two issues:
first, we only generate this many bits on expectation;
sometimes we generate less. We have to do the analysis
carefully to account for this. Second, to generate
a single bit of bias $\gamma$ still requires us to query
the noisy oracle with bias $\delta$ and pay the full
$\delta^2$; in other words, we do not necessarily
have the ability to amortize this cost. This can
happen once per bit.

To analyze the total expected cost, we start by
generating one bit of bias $\delta$ for each of the $n$
input positions, and using those bits to initiate
random walks that reach $t$ or $-t$. The cost
of this initiation phase is $n\delta^2=1$ (since
$\delta=1/\sqrt{n}$). Thereafter, we only
query the noisy oracle of bias $\delta$ when necessary,
that is, when we run out of the artificially-generated
$\gamma$-biased bits. The total expected cost of this
procedure is the sum of the expected cost for
each of the $n$ input positions, so we analyze the cost
of a single input position.

For such a position, what happens is that a walk
of bias $\gamma$ is generated, and then cut off in a way
that can depend on the walk so far as well as on independent
randomness. We know the expected number of steps
before cutting off is $T_i$
(where $\sum_i T_i=O(\noisyR(f)/\gamma^2)$),
and we wish to bound the expected number of bits of bias
$\delta$ we must generate to simulate this sequence --
which means we must bound the expected number of times the walk
crossed a point which is a multiple of $t$ (not counting
the same multiple of $t$ if it occurs twice in a row).
But each time we reach a multiple of $t$, it is effectively
as if we start back at $0$.

In other words, let $X$ be the
random variable for the number of steps it takes to reach
$t$ or $-t$ starting at $0$ (with bias $\gamma$).
We know that $\bE[X]=\mu$, where $\mu\ge\delta^2/120\gamma^2$.
We play the following game: we add up independent copies of $X$,
which we label $X_1,X_2,\dots$, and we stop adding them
by some stopping rule $L$ where $L$ is a random variable
that can depend on $X_1,X_2,\dots, X_{L-1}$
(but not on $X_t$ for $t\ge L$). We know that
$\bE[\sum_{\ell=1}^L X_{\ell}]\le T_i$, and we wish to upper
bound $\bE[L]$ by $T_i/\mu$. This is
what's known as Wald's equation, which can be shown as follows
(using $I_t$ to denote the indicator random variable
with $I_t=0$ if $t>L$ and $I_t=1$ otherwise):
\[\bE\left[\sum_{t=1}^L X_t\right]
\!=\bE\left[\sum_{t=1}^\infty X_tI_t\right]
\!=\!\sum_{t=1}^\infty\bE[X_tI_t]
=\!\sum_{t=1}^\infty\Pr[I_t=1]\bE[X_t|I_t=1]
=\!\sum_{t=1}^\infty\Pr[L\ge t]\bE[X_t]
=\mu\bE[L].\]
This line crucially uses the fact that $\bE[X_t|L\ge t]=\bE[X_t]$,
which holds because $L$ depends only on $X_1,X_2,\dots,X_{L-1}$
but not on $X_L$.
Thus we have $T_i\ge\bE[L]\mu$, or $\bE[L]\le T_i/\mu$.
Hence the expected number of queries to the $\delta$-biased
oracle is $T_i/\mu$, and summing over all $i$, it is
at most $O(\noisyR(f)/\gamma^2\mu)=O(\noisyR(f)/\delta^2)$.

The final cost of the algorithm is therefore $O(\noisyR(f))+1$. Since $\noisyR(f)=\Omega(1)$,
this is $O(\noisyR(f))$, as desired.
\end{proof}

We now prove \thm{noisyR_gapmaj},
showing that $\noisyR(f)=\Theta(\R(f\circ\gapmaj_n)/n)$
for every (possibly partial) Boolean functions $f$,
where $n$ is the input size of $f$.
We note that this theorem also holds for relations:
we have
\[\noisyR_\epsilon(f)
=\Theta(\R_\epsilon(f\circ\gapmaj_n)/n)\]
for any constant $\epsilon$\footnote{recall
that relations cannot be amplified, so $\epsilon$
matters.} and
for any relation $f$ that has two inputs $x,y$
with disjoint allowed output sets.

In one direction, this follows via
\thm{composition}: we have
\[\R(f\circ\gapmaj_n)=\Omega(\noisyR(f)\R(\gapmaj_n)),\]
and $\R(\gapmaj_n)=\Omega(n)$ by \lem{gapmaj}.
Hence $\noisyR(f)=O(\R(f\circ\gapmaj_n)/n)$,
even for relations $f$ (since \thm{composition}
holds for relations).

In the other direction, fix a function or relation
$f$. By \thm{sqrt_bias}, there is some noisy
oracle algorithm $A$ with worst-case expected
cost $O(\noisyR(f))$ that computes $f$
using only noisy oracle calls with parameter
$1$ or $1/\sqrt{n}$. We can easily
turn this into an algorithm for $f\circ\gapmaj_n$
whose cost is $O(n)$ times larger:
a noisy oracle call with parameter $1$ to a bit $x_i$
of the input to $f$ will
be implemented by querying the entire $\gapmaj_n$
gadget at that position, incurring a cost
of $n$ instead of $1$. On the other hand,
a noisy oracle call with parameter $1/\sqrt{n}$
to bit $x_i$
will be implemented by querying a single, random
bit of the corresponding $\gapmaj_n$ input.
This will incur cost $1$ instead of cost $1/n$.
Note that the bias of a single query to the
$\gapmaj_n$ input might be slightly different
than $1/\sqrt{n}$ due to rounding. If it's slightly
larger, we can simply add noise to get bias exactly
$1/\sqrt{n}$. If it's slightly smaller, we can
query several bits independently at random in order
to amplify the bias slightly, reducing to the case
where the bias is slightly larger than $1/\sqrt{n}$.
This costs only a constant factor overhead.
We conclude that $A$ can be converted to an algorithm
solving $\R(f\circ \gapmaj_n)$ which makes
$O(n\cdot\noisyR(f))$ queries, as desired.

\subsection{The non-adaptive case}

We now turn to the non-adaptive setting,
in order to show that in that
setting, $\noisyR(f)$ becomes equal to $\R(f)$.

\subsubsection*{Definitions}

First, we properly define the non-adaptive complexity measures $\noisyR^\na(f)$ and $\R^\na(f)$.
To start, a \emph{deterministic} non-adaptive
algorithm is a subset $S\subseteq [n]$
together with a map $\alpha\colon\B^S\to\B$; when we apply
such an algorithm $(S,\alpha)$ to an input $x\in\B^n$,
the output will be $\alpha(x_S)$, where $x_S$ denotes
the string $x$ restricted to the positions in $S\subseteq[n]$.
The cost of $(S,\alpha)$ will be $|S|$.

A \emph{randomized} non-adaptive algorithm will then
simply be a probability distribution over deterministic
non-adaptive algorithms, and for such a randomized algorithm $R$
we will let $\cost(R)$ be the expectation of $|S|$ and
we will let $\height(R)$ be the maximum value of $|S|$
for $(S,\alpha)$ in the support of $R$.
We let $R(x)$ denote the random variable which
takes value $\alpha(x_S)$ when $(S,\alpha)$ is sampled from $R$,
and we will say $R$ computes $f$ to worst-case error $\epsilon$
if $\Pr[R(x)\ne f(x)]\le \epsilon$ for all $x\in\Dom(f)$.
Then $\R^\na_\epsilon(f)$ will be the minimum height of
a non-adaptive randomized algorithm computing $f$, and
$\bar{\R}^\na_\epsilon(f)$
will be the minimum worst-case cost of such an algorithm.
As usual, we omit $\epsilon$ when it equals $1/3$,
and we note that $\bar{\R}^\na(f)=\Theta(\R^\na(f))$
due to Markov's inequality.

A $\noisyR^\na(f)$
algorithm will also be a probability distribution
over pairs $(S,\alpha)$, but this time $S$ will contain
a multiset of noisy queries instead of a set of queries.
The multiset $S$ will contain pairs $(i,\gamma)$ where
$i\in[n]$ and $\gamma\in(0,1]$; each such pair represents
a query to the noisy oracle for $x_i$ with parameter $\gamma$.
We will require the multiset $S$ to be finite, and we will
also require the probability distribution over
pairs $(S,\alpha)$ to have finite support.
A single query $(i,\gamma)$ will have cost $\gamma^2$,
the cost of $S$ will be the sum of the costs of its elements,
and the cost of a noisy algorithm will be the expected
cost of $S$ for over pairs $(S,\alpha)$ sampled from the algorithm.
The output of such an algorithm $R$ on input $x$,
denoted $R(x)$, will be the random variable corresponding
to sampling $(S,\alpha)$ from $R$, making noisy queries
to $x$ as specified by $S$, and applying the Boolean function
$\alpha$ to the result of those queries.
We say $R$ computes $f$ to error $\epsilon$
if $\Pr[R(x)\ne f(x)]\le\epsilon$ for all $x\in\Dom(f)$.
We define $\noisyR^\na_\epsilon(f)$ to be the infimum of $\cost(R)$
over noisy non-adaptive algorithms $R$ which compute $f$
to error $\epsilon$; when $\epsilon=1/3$, we omit it.

\subsubsection*{Some simplifications}

We observe that we can simplify $\noisyR^\na(f)$
substantially. First, we can remove pairs $(S,\alpha)$
from the domain of a noisy non-adaptive algorithm $R$
if the cost of $S$ is larger than $10$ times the cost of $R$;
using Markov's inequality, this only changes
the error of $R$ by an additive $1/10$, and we can amplify
this back to error $1/3$. This means the worst-case and
expected versions of $\noisyR^\na(f)$ are equivalent
up to constant factors. 
Second, we can add an additional query to each bit with
bias $1/\sqrt{n}$; this only increases the cost of the algorithm
by an additive $n\cdot (1/\sqrt{n})^2=1$. We observe
that by \lem{noisyR_one} we have $\noisyR^\na(f)=\Omega(1)$, so
this increase by an additive $1$ is only a constant factor
increase.
Next, inside each multiset $S$ of queries, we can
combine all the queries to bit $i$ into a single noisy
query with a larger bias parameter, so that the algorithm
only makes at most one noisy query to each bit $i$;
furthermore, for each $i$, the bias parameter will be at least
$1/\sqrt{n}$.

Finally, using arguments from \lem{two_biases}, we can
assume the noisy non-adaptive algorithm queries each bit
with bias parameter either $1/\sqrt{n}$ or $1$. If
the total cost of this noisy non-adaptive algorithm is $T$,
then it makes at most $T$ exact queries (with bias $1$)
and at most $nT$ noisy queries with parameter $1/\sqrt{n}$.
Hence we can split $S$ into $A$ and $B$,
where $A\subseteq[n]$ is a set of size $T$
representing exact queries
and $B$ is a multiset of elements from $[n]$ of size $nT$
representing noisy queries (with bias $1/\sqrt{n}$).
In other words, $\noisyR^\na(f)$ is (up to constant
factors) the minimum positive integer $T$ such that there is
a probability distribution $R$ over $(A,B,\alpha)$
with $|A|=T$ and $|B|=nT$ which computes $f$
to error $1/3$, where the output $R(x)$ is generated
by sampling $(A,B,\alpha)$, querying the $T$ bits in $A$,
making noisy queries of bias $1/\sqrt{n}$ to the $nT$
bits in $B$, feeding the results to $\alpha$, and returning
the bit $\alpha$ returns.

\subsubsection*{Switching to the distributional
setting}

In order to prove \thm{nonadaptive}, which states that
$\noisyR^\na(f)=\R^\na(f)$ for partial Boolean functions, we start with the following minimax lemma for
non-adaptive algorithms.

\begin{lemma}\label{lem:nonadaptive_minimax}
Let $f$ be a (possibly partial) Boolean function.
Then there is a distribution $\mu$ over $\Dom(f)$
such that any randomized non-adaptive algorithm
$R$ with $\height(R)<\R^\na_\epsilon(f)$ must
make average error greater than $\epsilon$ when run on inputs
from $\mu$.
\end{lemma}

\begin{proof}
This follows from a standard minimax argument dualizing
across the error. That is, let $\mathcal{R}$ be the set
of all randomized non-adaptive algorithms with height
less than $\R^\na_\epsilon(f)$, and let $\Delta$ be the
set of all probability distributions over $\Dom(f)$.
Then a standard minimax theorem gives
\[\adjustlimits\min_{R\in\mathcal{R}}\max_{\mu\in\Delta}
\Pr_{x\sim\mu}[R(x)\ne f(x)]
=\adjustlimits\max_{\mu\in\Delta}\min_{R\in\mathcal{R}}
\Pr_{x\sim\mu}[R(x)\ne f(x)],\]
since $\Pr_{x\sim\mu}[R(x)\ne f(x)]$ is bilinear as a function
of $\mu$ and of $R$. The left hand side is the worst-case error
of randomized non-adaptive algorithms of height less than
$\R^\na_\epsilon(f)$, which must be strictly greater than
$\epsilon$. The right hand side then provides a distribution
$\mu$ which is hard for all randomized algorithms of small height.
\end{proof}

We use this minimax lemma to switch to the distributional
setting. That is, we consider a noisy non-adaptive algorithm
$R$ that succeeds in the worst case, and show that for any
fixed distribution $\mu$, we can convert $R$ into
a non-noisy non-adaptive algorithm $R'$ which has similar
cost to $R$ and computes $f$ to bounded error against $\mu$.
Taking $\mu$ to be the hard distribution from
\lem{nonadaptive_minimax} will then give
$\R^\na(f)=O(\noisyR(f))$.

\subsubsection*{Defining a clean algorithm from a noisy one}

Let $R$ be a noisy non-adaptive algorithm which makes $T$
exact queries
and $nT$ noisy queries, with $T=O(\noisyR^\na(f))$,
and computes $f$ to error $1/1000$. Since $R$ computes $f$
to error $1/1000$ in the worst case, it also computes $f$
to error $1/1000$ against inputs from $\mu$; this means
there is some deterministic $(A,B,\alpha)$ in the support of $R$
which also computes $f$ to error at most $1/1000$ against $\mu$.

We will now define a non-adaptive randomized algorithm $R'$
as follows: $R'$ queries all the bits in $A$, and in addition
queries each bit in the multiset $B$ with probability $1/n$
(sampled independently). Then $R'$ computes the posterior
distribution $\mu'$ defined by starting with prior $\mu$
and doing a Bayesian update on the bits $R'$ has seen;
if $\mu'$ has more probability mass on $1$-inputs than $0$-inputs
of $f$, $R'$ then outputs $1$, otherwise $R'$ outputs $0$.

It's easy to see that $R'$ is a non-adaptive randomized algorithm
with expected number of queries equal to
$|A|+|B|/n=O(\noisyR^\na(f))$. It remains to show that $R'$
computes $f$ against $\mu$ to small error, say $1/10$;
then we can remove from the support of $R'$ the query sets
that are larger than $10$ times the expectation, and get
a randomized algorithm $R''$ which uses $O(\noisyR^\na(f))$
worst-case queries and still makes error at most $1/10+1/10<1/3$
against $\mu$. By the definition of $\mu$, this implies
that $\noisyR^\na(f)=\Omega(\R^\na(f))$, as desired.

\subsubsection*{Rephrasing the error analysis in terms of
noisy channels}

To analyze the error that $R'$ makes against $\mu$,
we first make the following modification to the strings
under consideration. For each $x\in\Dom(f)$, we define
the string $\hat{x}$ as the string of length $mT+nT$
whose first $mT$ bits are the bits of $x$ from $A$ copied $m$ times
each, and whose next $nT$ bits are the bits of $x$ from the multiset
$B$ (this will require duplicating bits of $x$ and
rearranging them). Since $A\cup B$ contains all bits of $x$
at least once, the resulting string $\hat{x}$ uniquely determines
the original string $x$ (but has many of its bits duplicated multiple
times). We can therefore modify the function
$f$ to get $\hat{f}$ such that $\hat{f}(\hat{x})=f(x)$,
and modify $\mu$ to get $\hat{\mu}$ over the modified strings.
Note that each non-adaptive randomized algorithm on the original
strings $x$ can be modified to work on the strings $\hat{x}$,
and vice versa. This also works for noisy randomized algorithms.
The parameter $m$ will be chosen to be much larger than $n$.

We wish to argue that $R'$ has high success probability, using
the fact that $R$ has good success probability. We note
that since $R$ succeeds with good probability, we can compute
$\hat{f}$ against $\hat{\mu}$ simply by making one noisy
query to each bit of the input $\hat{x}$, with parameter
$1/\sqrt{n}$ each. In other words, if $X$ is the random variable
with probability distribution $\hat{\mu}$, let $N_\gamma(\cdot)$
denote the \emph{noisy channel} where each bit of the string
gets flipped with independent probability $(1-\gamma)/2$.
Then we can compute $\hat{f}(X)$ by observing
only $Y=N_{1/\sqrt{n}}(X)$. To do so, we first compute the
bits in $A$ using $Y$: for each bit in $A$, we receive
$m$ noisy versions of it sampled independently with bias $1/\sqrt{n}$
each. Taking a majority vote of these noisy versions, and
assuming $m$ is much larger than $n$, we get an estimate for
each bit in $A$ which has error probability as small as we'd like.
Afterwards, we use these bits in $A$, combined with the noisy
bits in $B$, and apply $\alpha$ to get an estimate of
$f(x)=\hat{f}(X)$. The original error probability of $\alpha$
(against $\mu$) was $1/1000$; by picking $m$ large enough,
we can get this new protocol to have error probability at most
$1/999$. In other words, we have a function $\beta$
such that $\Pr[\beta(Y)\ne\hat{f}(X)]\le 1/999$,
where $Y=N_{1/\sqrt{n}}(X)$.

\subsubsection*{Switching from error probability to
relative entropy}

This gave us a noisy channel way to express the success probability
of $R$. We now express the success probability of $R'$
in terms of an \emph{erasure} channel. That is,
let $E_\gamma(X)$ denote the channel that replaces each bit
of $X$ with $*$ except with independent probability $\gamma$.
Consider the string $E_{1/n}(X)$. This string
erases each bit of $X$ with probability $1-1/n$, and keeps it
with probability $1/n$. For each bit that was originally in $A$,
there are $m$ copies of this bit in $X$, so the probability that
all the copies get erased can be made arbitrarily small (by
picking $m$ large enough). On the other hand, each bit that
was originally in the multiset $B$ is only kept with probability
$1/n$. Hence the string $Z=E_{1/n}(X)$ has distribution
arbitrarily close to the distribution of queries made by $R'$
against $\mu'$. Therefore, it suffices to prove that
$\Pr[\beta'(Z)\ne \hat{f}(X)]<1/11$ where $\beta'$ is the function
that selects the best Bayesian guess for $\hat{f}(X)$ given
observation $Z=E_{1/n}(X)$.

We now use the following lemma to rephrase our goal in
information-theoretic terms.

\begin{lemma}\label{lem:error_entropy}
For any random variables $X$ and $Y$ on supports
$\mathcal{X}$ and $\mathcal{Y}$ respectively, and for
any functions $f \colon\mathcal{X} \to \{0,1\}$
and $\beta \colon \mathcal{Y} \to \{0,1\}$, we have
\[H( f(X) \mid Y) \le h\big( \Pr[ f(X) \neq \beta(Y)] \big).\]
Moreover, for all such $X$, $Y$, and $f$, there exists a function
$\beta$ such that
\[2\Pr[f(X)\ne \beta(Y)]\le H(f(X)\mid Y),\]
where $h$ is the binary entropy function. In particular,
$\beta$ can be chosen to be the Bayesian posterior function
for guessing $f(X)$ using $Y$.
\end{lemma}

The upper bound is a special case of Fano's inequality, and the lower bound was established by Hellman and Raviv~\cite{HR70}. We include the (easy) proof of the lemma for completeness.

\begin{proof}
Using Jensen's inequality,
\begin{align*}
H( f(X) \mid Y)
&= \E_y\big[ h(\Pr[ f(X) \neq \beta(Y) \mid Y = y])\big] \\
&\le h\big( \E_y[\Pr[ f(X) \neq \beta(Y) \mid Y = y]] \big) \\
&= h\big( \Pr[ f(X) \neq \beta(Y) ]\big).
\end{align*}
In the other direction,
using the fact that $2x \le h(x)$ for each $0 \le x \le \frac12$,
we have
\begin{align*}
2\Pr[ f(X) \ne \beta(Y)]
&= \E_{y}\big[ 2\Pr[ f(X) \neq \beta(Y) \mid Y = y] \big] \\
&\le \E_y\big[ h(\Pr[ f(X) \neq \beta(Y) \mid Y = y])\big] \\
&= \E_y\big[ h(\Pr[ f(X) = 1 \mid Y = y])\big]
= H( f(X) \mid Y).
\end{align*}
Note that in the second line, we used
$\Pr[f(X)\ne\beta(Y)\mid Y=y]\le1/2$,
which follows from our choice of $\beta$.
\end{proof}

Using this lemma, we get that
$H(\hat{f}(X)\mid Y)\le h(1/999)\le 1/87$
where $Y=N_{1/\sqrt{n}}(X)$, and we wish to show that
$H(\hat{f}(X)\mid Z)\le 1/22$, where $Z=E_{1/n}(X)$.
In particular, it suffices to show that
$H(\hat{f}(X)\mid Z)\le H(\hat{f}(X)\mid Y)$.

\subsubsection*{Appealing to a theorem of Samorodnitsky}
\label{sec:Sam}

To finish the proof, all we need is a special case of an inequality of
Samorodnitsky~\cite{Sam16} established by Polyanskiy and Wu~\cite{PW17}.
(See \app{SamProof} for more details.)

\begin{restatable}[Samorodnitsky~\cite{Sam16,PW17}]{theorem}{NoisyEntropy}
\label{thm:noisyEntropy}
For any function $f : \{0,1\}^n \to \{0,1,*\}$, any distribution $\mu$ on $f^{-1}(0) \cup f^{-1}(1)$, and any $0 \le \rho \le 1$, variables
$X \sim \mu$, $Y \sim N_\rho(X)$, and $Z \sim E_{\rho^2}(X)$ satisfy
\[
H( f(X) \mid Y ) \ge H( f(X) \mid Z ).
\]
\end{restatable}

This theorem says that adding noise (leaving bias $\rho$)
to $X$ preserves more information than erasing (leaving
the bit untouched with probability $\rho^2$).
It is exactly what we need to complete the proof, showing
that $\R^\na(f)=O(\noisyR^\na(f))$ for all partial functions
$f$. (The other direction, $\noisyR^\na(f)\le\R^\na(f)$,
follows directly from the definitions.)

As previously noted, this result $\R^\na(f)=O(\noisyR^\na(f))$
is false when $f$ is a relation. The step that fails is
the step where we switched from error probability
to relative entropy, in \lem{error_entropy};
this step has no clear analogue for relations.

\section*{Acknowledgements}

S.~B.~thanks Aditya Jayaprakash for collaboration on related research questions during the early stages of this project.
We thank Andrew Drucker, Mika G{\"o}{\"o}s, and Li-Yang Tan for correspondence about their ongoing work~\cite{BDG+20}.

\newpage
\appendix

\section{Amplifying small biases}\label{app:amplify}

In this appendix, we prove \lem{amplify}, which we restate
below.

\amplify*

To prove this lemma, we will require bounds on the mean
absolute deviation of the binomial distribution with
parameter $p=1/2$. Recall that the mean absolute deviation
is the expectation of $|X-\bE[X]|$, where $X$ is a random
variable (which for us will have a binomial distribution).

\begin{lemma}
The mean absolute deviation $M_k$ of the binomial distribution
with parameters $k$ and $1/2$ (where $k$ is an odd integer)
satisfies
\[\sqrt{\frac{k}{2\pi}}\le M_k
\le \sqrt{\frac{k}{2\pi}}\left(1+\frac{1}{k}\right).\]
\end{lemma}

\begin{proof}
A closed form expression for the mean absolute deviation
of the binomial distribution with parameters $1/2$ and $k$
(where $k$ is odd) is known (see, for example, \cite{DZ91}):
\[M_k=2^{-k}\left(\frac{k+1}{2}\right)\binom{k}{(k-1)/2}.\]
To prove the result, we only need to bound the binomial
coefficient above sufficiently accurately.
We know that
\[\binom{k}{(k-1)/2}=r_k\sqrt{\frac{2}{\pi k}}2^k,\]
where $r_k$ is an error term close to $1$.
To prove the desired bounds, we need only show that
$r_k\ge k/(k+1)$ and $r_k\le 1$.

From \cite{Sta01} (Corollary 2.4, setting $n=1$, $m=k$, $p=(k-1)/2$),
we get
\[r_k=\alpha_k\left(1+\frac{1}{k^2-1}\right)^{k/2}\left(1-\frac{1}{k+1}\right),\]
where $\alpha_k$ satisfies
\[e^{1/12k-1/(6k-6)-1/(6k+6)}<\alpha_k<e^{1/12k-1/(6k-5)-1/(6k+7)}.\]
Note that using $k\ge 3$, we get $\alpha_k>e^{-7/24k}>e^{-1/3k}$,
and for all $k\ge 7$ (as well as checking $k=3,5$ by hand)
we get $\alpha_k<e^{-1/4k}$. Using $e^{x/(1+x)}<1+x$, we get
the lower bound
\[r_k>e^{-7/24k}e^{1/2k}e^{-1/k}=e^{-19/24k}>1-19/24k>1-5/6k
=1-1/(k+k/5)\ge1-1/(k+1)\]
assuming $k\ge 5$. For $k=3$, we can calculate $r_3$ and check
it is larger than $3/4$, so $r_k>k/(k+1)$ for all $k\ge 3$.

For the upper bound, we use $k\ge 3$ to get
\[r_k<e^{-1/4}e^{9/16k}e^{-3/4k}=e^{-7/16k}<1.\]
Finally, the case $k=1$ can be verified directly, as $M_k=1/2$
in that case.
\end{proof}

Next, we note that it is clear $\gamma'$ and $\gamma$
have the same sign, and that the cases $\gamma>0$ and
$\gamma<0$ are symmetric. For this reason, we can
restrict to the $\gamma>0$ case without loss of generality.
We note that $\gamma'$ is the probability of $X=1$
minus the probability of $X=0$, so we have
\begin{align*}
\gamma'&=\sum_{i=(k+1)/2}^{k}\binom{k}{i}\left(\frac{1+\gamma}{2}\right)^i\left(\frac{1-\gamma}{2}\right)^{k-i}-\sum_{i=0}^{(k-1)/2}\binom{k}{i}\left(\frac{1+\gamma}{2}\right)^i\left(\frac{1-\gamma}{2}\right)^{k-i}\\
&=\sum_{i=0}^{(k-1)/2}\binom{k}{i}\left[\left(\frac{1+\gamma}{2}\right)^{k-i}\left(\frac{1-\gamma}{2}\right)^i-\left(\frac{1+\gamma}{2}\right)^i\left(\frac{1-\gamma}{2}\right)^{k-i}\right]\\
&=2^{-k}\sum_{i=0}^{(k-1)/2}\binom{k}{i}(1-\gamma^2)^i
[(1+\gamma)^{k-2i}-(1-\gamma)^{k-2i}].  
\end{align*}

\subsection{The lower bound}

Note that $(1+\gamma)^x-(1-\gamma)^x\ge 2\gamma x$ for all
$\gamma\in[0,1/3]$ and all positive integer $x$.
To see this, observe that
they are equal when $\gamma=0$, and the derivative of the
left hand side (with respect to $\gamma$) is
$x(1+\gamma)^{x-1}+x(1-\gamma)^{x-1}$, which we just need to show
is larger than $2x$ for positive integer $x$. This clearly holds
for $x=1$ and $x=2$, so suppose $x\ge 3$. It suffices to show
$(1+\gamma)^{x-1}-1\ge 1-(1-\gamma)^{x-1}$. The two sides
are equal at $\gamma=0$, and when $\gamma> 0$, the derivative
of the left is larger than that of the right. Hence the inequality
holds.

Together with $(1-\gamma^2)^i\ge (1-\gamma^2)^{k/2}\ge 1-k\gamma^2/2$,
this gives us
\[\gamma'\ge 2^{1-k}(1-k\gamma^2/2)\gamma
\sum_{i=0}^{(k-1)/2}\binom{k}{i}(k-2i)=2\gamma M_k(1-k\gamma^2/2).\]
Using $M_k\ge\sqrt{k/2\pi}$, we get
\[\gamma'\ge \sqrt{\frac{2}{\pi}}\sqrt{k}\gamma(1-\gamma^2k/2).\]
Finally, since $k\le1/\gamma^2$, we get
\[\gamma'\ge\frac{1}{\sqrt{2\pi}}\sqrt{k}\gamma \ge \frac13\sqrt{k}\gamma.\]

\subsection{The upper bound}

We have for any real number $a$ between $0$ and $(k-1)/2$,
\[\gamma'\le 2^{-k}\sum_{i=0}^{a}\binom{k}{i}(1-\gamma^2)^i
[(1+\gamma)^{k-2i}-(1-\gamma)^{k-2i}]
+2^{-k}\sum_{i=a}^{(k-1)/2}\binom{k}{i}(1-\gamma^2)^i
[(1+\gamma)^{k-2i}-(1-\gamma)^{k-2i}],\]
where if $a$ is not an integer the former sum ends at its floor
and the latter starts at its ceiling. We upper bound these two
sums separately (and choose $a$ later).
Denote the first sum by $S_1$ and the second by $S_2$.

For $S_1$ we omit the $(1-\gamma)^{k-2i}$
term and simplify, writing
\[S_1\le \sum_{i=0}^a\binom{k}{i}
\left(\frac{1-\gamma}{2}\right)^i\left(\frac{1+\gamma}{2}\right)^{k-i}.\]
This is the probability that a Binomial random variable with parameters
$(1-\gamma)/2$ and $k$ is at most $a$. Using the Chernoff bound,
we get
\[S_1\le e^{-((1-\gamma)k-2a)^2/2}.\]

To upper bound $S_2$, we bound $(1-\gamma^2)^i$ by $1$,
and we write
\[(1+\gamma)^{k-2i}-(1-\gamma)^{k-2i}=\sum_{\ell=0}^{(k-1)/2-i}
\binom{k-2i}{2\ell+1}2\gamma^{2\ell+1}
\le 2\gamma(k-2i)\sum_{\ell=0}^{(k-1)/2-i} \frac{(\gamma(k-2i))^{2\ell}}{(2\ell+1)!}\]
\[\le 2\gamma(k-2i)\sum_{\ell=0}^{(k-1)/2-i} \frac{(\gamma(k-2i))^{2\ell}}{\ell!\;6^\ell}
\le 2\gamma(k-2i)e^{\gamma^2(k-2i)^2/6}.\]
Hence we have
\[S_2\le2\gamma e^{\gamma^2(k-2a)^2/6} 2^{1-k}
\sum_{i=a}^{(k-1)/2}\binom{k}{i}\left(\frac{k}{2}-i\right).\]
Note that
\[2^{1-k}\sum_{i=a}^{(k-1)/2}\binom{k}{i}\left(\frac{k}{2}-i\right)
\le 2^{1-k}\sum_{i=0}^{(k-1)/2}\binom{k}{i}\left(\frac{k}{2}-i\right)
=M_k\le\sqrt{\frac{k}{2\pi}}\left(1+\frac{1}{k}\right)
\le(3/5)\sqrt{k}\]
for $k\ge 3$. Thus, for $k\ge 3$, we have
\[\gamma'=S_1+S_2\le e^{-((1-\gamma)k-2a)^2/2}+(6/5)\gamma\sqrt{k} e^{\gamma^2(k-2a)^2/6}.\]
Recall that $a$ was arbitrary. Picking
$a= (1-\gamma)k/2-\sqrt{(1/2)\ln(1/\gamma)}$
will cause the first term above to be equal to $\gamma$. The
second term to become $(6/5)\sqrt{k}\gamma$ times
$e^{\gamma^2(\gamma k+\sqrt{2\ln(1/\gamma)})^2/6}$.
Using $(y+z)^2\le 2y^2+2z^2$, this last part is at most
$e^{(\gamma^4 k^2+2\gamma^2\ln(1/\gamma))/3}$.
Using $\gamma^4 k^2\le 1$ and
$2\gamma^2\ln(1/\gamma))/3\le (2\ln 3)/27$,
this expression evaluates to at most $1.6$,
and we get
\[\gamma'\le \gamma+2\sqrt{k}\gamma\le 3\sqrt{k}\gamma.\]

\section{On Samorodnitsky's theorem}
\label{app:SamProof}

\thm{noisyEntropy} as stated in \sec{Sam} is not found explicitly in~\cite{Sam16} but it follows directly from the following variant of the theorem as established by Polyanskiy and Wu~\cite{PW17}.

\begin{theorem}[Theorem 20 in~\cite{PW17}]
\label{thm:PWS}
Consider the Bayesian network 
\[
U \to X^n \to Y^n,
\]
where $P_{Y^n|X^n} = \prod_{i=1}^n P_{Y_i|X_i}$ is a memoryless channel with $\eta_i := \eta_{\mathrm{KL}}(P_{Y_i|X_i})$. Then we have
\[
I(U; Y^n) \le I(U; X_S \mid S) = I(U; X_S, S),
\]
where $S \perp\!\!\!\!\perp (U,X^n,Y^n)$ is a random subset of $[n]$ generated by independently sampling each element $i$ with probability $\eta_i$.
\end{theorem}

For completeness, we show how \thm{PWS} implies \thm{noisyEntropy}, restated below.

\NoisyEntropy*

\begin{proof}
Fix any partial function $f : \{0,1\}^n \to \{0,1,*\}$, any distributions $\mu_0$ on $f^{-1}(0)$ and $\mu_1$ on $f^{-1}(1)$, and parameter $p \in [0,1]$. Let $\mu = p \mu_0 + (1-p)\mu_1$.

Define $U$ to be the random variable on $\{0,1\}$ for which $\Pr[ U = 0 ] = p$. Define $X$ to be a random variable drawn from $\mu_U$. And define $Y = N_\rho(X)$ to be the random variable obtained by applying the noise operator independently to each coordinate of $X \in \{0,1\}^n$. Then $U, X, Y$ satisfy the conditions of Theorem~\ref{thm:PWS} and the identity $U = f(X)$, so
\[
I( f(X); Y ) \le I( f(X); X_S \mid S) = I( f(X) ; Z)
\]
when $Z$ is obtained from $X$ by erasing each coordinate of $X$ independently with probability $1-\rho^2$. Thus,
\[
H( f(X) \mid Y) = H(f(X)) - I( f(X); Y) \ge H(f(X)) - I( f(X); Z) = H(f(X) \mid Z). \qedhere
\]
\end{proof}

\newpage

\end{fulltext}

\iffocs
\else
\phantomsection\addcontentsline{toc}{section}{References} 
\renewcommand{\UrlFont}{\ttfamily\small}
\let\oldpath\path
\renewcommand{\path}[1]{\small\oldpath{#1}}
\fi
\emergencystretch=1em 
\printbibliography

\end{document}